\documentclass[journal]{IEEEtran}
\usepackage{amsmath,amsfonts, amssymb}
\usepackage{algorithmic}
\usepackage{algorithm}
\usepackage{array}
\usepackage[caption=false,font=normalsize,labelfont=sf,textfont=sf]{subfig}
\usepackage{textcomp}
\usepackage{stfloats}
\usepackage{url}
\usepackage{verbatim}
\usepackage{graphicx}
\usepackage{cite}
\usepackage{amssymb}
\usepackage{amsmath}
\usepackage{hyperref}
\usepackage{graphicx}
\usepackage{subfig}
\usepackage{booktabs}
\usepackage{tabularx}
\usepackage{array}
\usepackage{amsthm}
\newtheorem{definition}{Definition}
\newtheorem{assumption}{Assumption}
\newtheorem{lemma}{Lemma}
\newtheorem{proposition}{Proposition}

\newtheorem{remark}{Remark}

\begin{document}

\title{
Telling Neural Network What We Know: Weight Configuration for RNN and Its Application to NextG Real-Time Neural Receiver
}

\author{Ummay Sumaya Khan, Lizhong Zheng, Lingjia Liu \vspace{-0.3in} 

\thanks{U. Khan, and L. Liu are with Wireless@Virginia Tech, the Bradley Department of Electrical and Computer Engineering, Virginia Tech. L. Zheng is with the Department of Electrical Engineering and Computer Science, Massachusetts Institute of Technology, Cambridge, MA 02139 USA.

The work is supported in parts by the US National Science Foundation under grants CNS-2003059 and ECCS-2128594. The corresponding author is Lingjia Liu (ljliu@ieee.org).}
}

\maketitle

\begin{abstract}
AI-enabled receivers offer a promising approach to adaptive physical-layer processing, but conventional neural receivers often require extensive offline training and computationally expensive backpropagation, limiting their suitability for low-latency cellular systems. Moreover, classical MIMO-OFDM receivers rely on accurate channel estimation, leading to substantial pilot overhead. This paper introduces \textbf{skip-RNN}, a lightweight skip-augmented recurrent neural receiver that performs online time-domain receive processing using fixed recurrent and skip branches with an adaptive linear readout. We show that the skip-RNN can be interpreted as a structured deconvolution filter whose recurrent and skip branches realize IIR and FIR basis functions, respectively. We develop a domain-knowledge informed weight-configuration framework for the skip-RNN by exploiting its structural correspondence with the target stable and causal inverse filter of mixed-phase channels. The dominant IIR structure extracted from estimated channel statistics is used to configure the recurrent dynamics, while the non-minimum-phase component is addressed through the skip branch and a statistical FIR readout prior incorporated into prior-centered regularized online readout adaptation. We further derive an approximation-error bound that characterizes the resulting configured neural-receiver representation. Simulations over standardized 3GPP channels demonstrate the effectiveness of the configured neural receiver under extremely low pilot overhead.
\end{abstract}

\begin{IEEEkeywords}
Recurrent neural network (RNN), Echo state network (ESN), MIMO-OFDM receiving process, non-line-of-sight (NLoS) channel, non-minimum phase, skip-RNN.
\end{IEEEkeywords}

\section{Introduction}
\label{sec:intro}

Multiple-input multiple-output (MIMO) with orthogonal frequency division multiplexing (OFDM) is a key technology in modern wireless systems, but reliable symbol detection remains challenging under multipath fading and imperfect channel knowledge. Conventional receivers typically employ linear minimum mean square error (LMMSE) detection using channel state information (CSI) estimated from demodulation reference signals (DMRS)~\cite{3gpp_mmse}. Their performance, however, can degrade substantially when channel estimates are inaccurate, particularly at low signal-to-noise ratios (SNRs).

Machine learning has therefore been widely explored for joint channel estimation and MIMO detection, including multi-layer perceptrons (MLPs)~\cite{ye2017power}, recurrent and long short-term memory (LSTM) networks~\cite{liao2019deep,lyu2020deep}, convolutional neural networks (CNNs)~\cite{honkala2021deeprx,zhao2021deep}, and generative adversarial networks (GANs)~\cite{yi2020deep}. Model-driven architectures such as OAMP-Net~\cite{he2018model}, DetNet~\cite{samuel2019learning}, and MMNet~\cite{khani2020adaptive} further incorporate trainable parameters into classical iterative detectors. Although recurrent neural networks are well suited to sequential signal processing~\cite{goodfellow2016deep}, conventional RNN training relies on backpropagation through time, resulting in substantial computational and data requirements~\cite{letter_ummay}.

These requirements limit the practicality of learning-based receivers in rapidly varying wireless environments. Channel conditions, spatial rank, scheduling, and interference may change within millisecond-scale TTIs~\cite{liu_magazine}, requiring fast online adaptation with limited over-the-air training samples. At the same time, excessive reference-signal overhead reduces spectral efficiency~\cite{3gpp_pilot}. Moreover, purely data-driven receivers often relearn known system structure, increasing complexity and reducing interpretability. Incorporating available domain knowledge into the receiver design can therefore reduce training overhead, improve generalization, and enable more efficient online adaptation~\cite{xu2023detect}.

To reduce the training burden of conventional recurrent networks, reservoir-based receivers such as Echo State Networks (ESNs)~\cite{jaeger2001echo} use fixed input and recurrent weights and train only a linear readout online. This substantially lowers training complexity and data requirements. Windowed Echo State Networks (WESNs)~\cite{zhou2020learning,zhou2020rcnet,xu2022rc,xu2023detect} further augment the reservoir with windowed input buffers and perform detection at the subframe or TTI level, making them well suited to rapidly varying wireless systems with limited pilot resources. Our recent work has also shown that reservoir-based receivers can outperform conventional RNNs under sparse pilot allocation.

Despite these advantages, existing WESN receivers remain largely black-box designs with limited interpretation of their underlying filtering structure. Prior work~\cite{jere2024universal} configured ESN recurrent weights using the target first-order linear time-invariant (LTI) system, while~\cite{jere2023channel,10974467} showed that channel-aware reservoir configuration can achieve optimal BER for strictly minimum-phase MIMO channels under perfect channel-statistics knowledge. However, principled weight configuration for skip-augmented reservoir receivers operating over realistic non-minimum-phase (NMP), non-line-of-sight (NLoS) channels with estimated channel remains unexplored.

Motivated by these limitations, we introduce \emph{skip-RNN}, a lightweight and training efficient neural receiver, that combines fixed recurrent dynamics with a finite window of received samples connected directly to a trainable linear readout. We first establish a signal-processing interpretation of skip-RNN architecture. For mixed-phase channels, the stable causal LMMSE filter can be represented, for a finite decision delay, by a parallel combination of an IIR component and a finite-length FIR component. The skip-RNN realizes the same structure: the recurrent branch generates stable IIR basis functions through its recurrent modes, while the skip branch provides a complete finite-length causal FIR delay basis. For MIMO-OFDM receive processing, the online readout then determines a linear combination of these fixed basis functions to approximate the target causal channel inverse. This correspondence provides a principled explanation for the suitability of the skip-RNN receiver for MIMO-OFDM over NLoS mixed-phase channels, rather than treating the neural receiver architecture as a purely black-box learning model.

The time-domain interpretation of the skip-RNN shows that it operates as a causal MIMO deconvolution filter whose impulse response consists of an infinite-duration recurrent component and a finite-duration skip component. This structured representation motivates configuring its fixed basis functions using available channel-domain knowledge rather than relying on conventional random reservoir initialization. Taking the causal LMMSE receiver as the target function, statistics of the estimated LMMSE inverse ensemble can be exploited to construct basis functions that are better aligned with the target function space before online readout adaptation. Specifically, the dominant low-dimensional structure of the causal IIR ensemble is extracted from the statistics and approximated using stable rational functions, whose poles and associated coefficients are used to configure the recurrent and input weights. In parallel, the FIR component is used to determine the skip-window length and construct a statistical FIR readout prior, which is incorporated through prior-centered regularized least-squares (LS) adaptation within each TTI. Consequently, only the linear readout is adapted online, substantially reducing the training complexity compared with fully trained recurrent receivers. The introduced neural receiver is evaluated over standardized 3GPP NLoS channel models, where multipath propagation produces mixed-phase channel responses, under a limited pilot overhead. The main contributions of this work are summarized as follows:

\begin{itemize}
\item We develop an interpretable signal-processing framework for time-domain causal inverse filter of mixed-phase channels. Using Wiener--Hopf factorization and finite decision delay, we show that the causal inverse filter admits an equivalent parallel IIR--FIR representation, consisting of stable rational IIR modes and a finite causal FIR component. This representation provides the analytical basis for the neural receiver architecture and the weight configuration process using prior channel statistics.

\item We formulate the skip-RNN neural receiver as a structured time-domain filtering architecture and show that its recurrent and skip branches naturally realize IIR and FIR basis functions. A geometric interpretation further shows that the fixed basis can be configured with prior channel knowledge based on the causal inverse filter structure, while online readout training projects the target receiver response onto the configured IIR--FIR function space.

\item We introduce a channel-informed weight-configuration framework that exploits estimated channel statistics as domain knowledge. The dominant low-dimensional structure is extracted through eigenspace decomposition and approximated by stable rational functions to configure the recurrent poles and input weights. The FIR component is used to determine the skip-window length and construct a statistical FIR readout prior, which is incorporated through prior-centered regularized LS adaptation per TTI. In addition, we derive an approximation-error bound for the configured neural receiver, providing a theoretical characterization of its representation capability.

\item We evaluate the configured skip-RNN over standardized 3GPP NLoS MIMO-OFDM channel models. The results demonstrate improved BER performance over the model-based LMMSE and sphere-decoding receivers using estimated CSI, as well as the learning-based SDNet and randomly initialized WESN baselines.
\end{itemize}

The remainder of this paper is organized as follows. Section~\ref{sec:preliminaries} introduces the MIMO-OFDM receiver model and derives the parallel IIR-FIR representation of the causal LMMSE inverse filter. Section~\ref{sec:online_skip_rnn} presents the skip-RNN neural receiver from signal-processing and geometric perspectives and formulates the readout approximation error as a weight-configuration problem. Section~\ref{sec:weight_configuration} develops the channel-informed weight-configuration method and derives the corresponding approximation-error bound. Section~\ref{sec:perform_eval} describes the simulation setup and evaluates the skip-RNN neural receiver over standardized 3GPP channel models. Finally, Section~\ref{sec:conclusion} concludes the paper.

\textbf{Notations:} 
Throughout this paper, subscripts $(\cdot)_{t}$ and $(\cdot)_{f}$ indicate time domain and frequency domain signals, respectively. For all other variables, non-bold letters denote scalars, bold lowercase letters denote vectors, and bold uppercase letters denote matrices. $\mathbb{R}$ and $\mathbb{C}$ represent the real and complex number spaces respectively and other sets are represented by calligraphic letters. The operator $(\cdot)^{H}$ represents the conjugate transpose, $(\cdot)^{\dagger}$ denotes the paraconjugate and \(*\) denotes discrete-time convolution. The quantities corresponding to skip-RNN are indexed by the $(\cdot)_\mathrm{SR}$.


\section{MIMO-OFDM Receiver Processing for Mixed Phase Wireless Channels}
\label{sec:preliminaries}

This section first summarizes the general process for MIMO-OFDM receive processing and then introduces the stable and causal receive processing for the channel-inversion framework. As illustrated in Fig.~\ref{fig:mimo_detect}, the received signal is processed using either a learning-based recurrent neural receiver or a conventional model-based LMMSE detector. Since the neural receiver approximates the underlying channel-inversion problem, we also demonstrate how the desired inverse filtering operation can be represented in a stable and causal form, particularly for non-minimum-phase channels. This representation provides the basis for understanding how the channel inverse domain knowledge can subsequently be utilized to configure the neural receiver.

\subsection{General Framework for MIMO-OFDM Receive Processing}
\label{subsec:mimo_ofdm_deconvolution}

\begin{figure}[!ht]
    \centering
    \includegraphics[width=\linewidth]{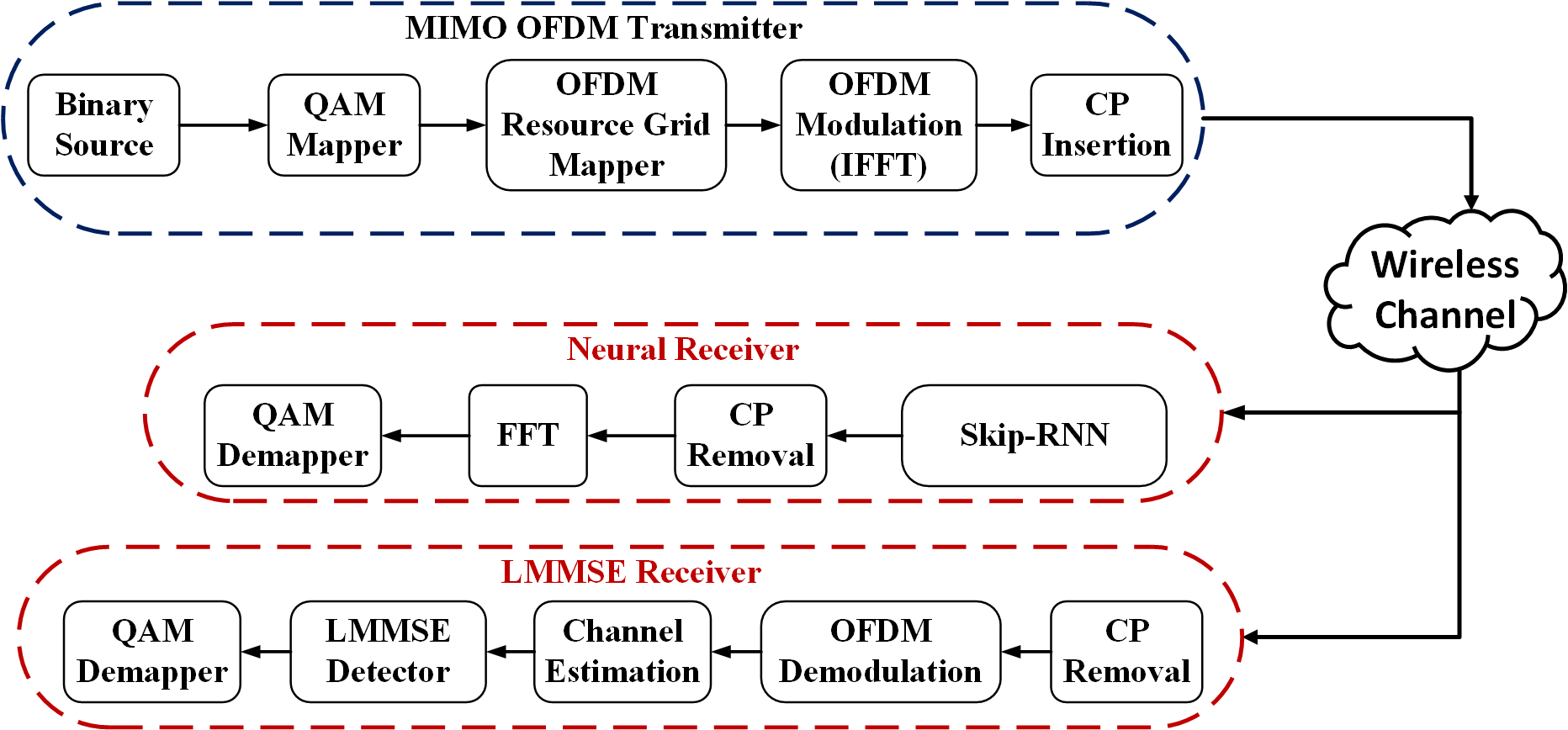}
    \caption{MIMO-OFDM receiver processing with conventional LMMSE and recurrent-neural-network detection branches.}
    \vspace{-0.8cm}
    \label{fig:mimo_detect}
\end{figure}

Consider a frequency-selective MIMO channel with \(N_t\) transmit and \(N_r\) receive antennas and an \(L\)-tap matrix-valued impulse response
\(\{\mathbf{H}[\ell]\}_{\ell=0}^{L-1}\), where \(\mathbf{H}[\ell]\in\mathbb{C}^{N_r\times N_t}\).
Let \(\mathbf{x}_t[\tau]\in\mathbb{C}^{N_t}\) and \(\mathbf{y}_t[\tau]\in\mathbb{C}^{N_r}\) denote the transmitted and received time-domain sample vectors, respectively. The received signal is

\begin{equation}
\label{eq:mimo_linear_convolution}
\mathbf{y}_t[\tau]
=
\sum_{\ell=0}^{L-1}
\mathbf{H}[\ell]\mathbf{x}_t[\tau-\ell]
+
\mathbf{v}_t[\tau]
\equiv
\left(\mathbf{H}*\mathbf{x}_t\right)[\tau]
+
\mathbf{v}_t[\tau],
\end{equation}
where \(\mathbf{v}_t[\tau]\in\mathbb{C}^{N_r}\) is additive noise.

CP-OFDM restructures this deconvolution into a blockwise form. We consider \(N_{\mathrm{fft}}\) subcarriers, of which
\(N_{\mathrm{sc}}\leq N_{\mathrm{fft}}\) are active, with
\(N_{\mathrm{sym}}\) OFDM symbols per TTI. For the \(n\)-th symbol, the frequency-domain transmit vectors
\(\mathbf{x}_f[n,q]\in\mathbb{C}^{N_t}\),
\(q=0,\ldots,N_{\mathrm{fft}}-1\), are mapped to the subcarriers and transformed by an \(N_{\mathrm{fft}}\)-point IDFT applied independently at each transmit antenna, producing the time-domain samples
\(\mathbf{x}_t[n,p]\in\mathbb{C}^{N_t}\).

A CP of length \(N_{\mathrm{cp}}\geq L-1\) is then appended. Assuming the channel remains approximately invariant over an OFDM symbol, the CP prevents inter-symbol interference in the retained samples and, after CP removal, converts the linear convolution in \eqref{eq:mimo_linear_convolution} into an \(N_{\mathrm{fft}}\)-point \emph{circular} convolution within each symbol~\cite{oppenheim1999discrete},
\begin{equation}
\label{eq:mimo_circular_convolution}
\mathbf{y}_t[n,p]
=
\sum_{\ell=0}^{L-1}
\mathbf{H}[\ell]
\mathbf{x}_t
\left[n,(p-\ell)_{N_{\mathrm{fft}}}\right]
+
\mathbf{v}_t[n,p],
\end{equation}
where \((\cdot)_{N_{\mathrm{fft}}}\) denotes modulo-\(N_{\mathrm{fft}}\) indexing.

Defining the CP-removed transmit and receive blocks as
\[
\overline{\mathbf{x}}_t[n]
=
\operatorname{col}
\left\{
\mathbf{x}_t[n,0],\ldots,
\mathbf{x}_t[n,N_{\mathrm{fft}}-1]
\right\},
\]
and
\[
\overline{\mathbf{y}}_t[n]
=
\operatorname{col}
\left\{
\mathbf{y}_t[n,0],\ldots,
\mathbf{y}_t[n,N_{\mathrm{fft}}-1]
\right\},
\]
the received OFDM block can be written compactly as
\[
\overline{\mathbf{y}}_t[n]
=
\widetilde{\mathbf{H}}_c
\overline{\mathbf{x}}_t[n]
+
\overline{\mathbf{v}}_t[n],
\]
where \(\widetilde{\mathbf{H}}_c\in\mathbb{C}^{N_rN_{\mathrm{fft}}\times N_tN_{\mathrm{fft}}}\) is the block-circulant matrix generated by \(\{\mathbf{H}[\ell]\}_{\ell=0}^{L-1}\).
Recovering \(\overline{\mathbf{x}}_t[n]\) from \(\overline{\mathbf{y}}_t[n]\) is thus a blockwise \emph{circular deconvolution}. Accordingly, time-domain receive processing can be written as
\[
\widehat{\overline{\mathbf{x}}}_t[n]
=
\widetilde{\mathbf{G}}_c[n]
\overline{\mathbf{y}}_t[n],
\]
where \(\widetilde{\mathbf{G}}_c[n] \in\mathbb{C}^{N_tN_{\mathrm{fft}}\times N_rN_{\mathrm{fft}}}\) is the block-circulant matrix corresponding to the time-domain sequence obtained from the inverse DFT of \(\{\mathbf{G}_f[n,q]\}_{q=0}^{N_{\mathrm{fft}}-1}\). Circular convolution is diagonalized by the DFT. 

The MIMO channel frequency response on subcarrier \(q\) is
\[
\mathbf{H}_f[q]
=
\sum_{\ell=0}^{L-1}
\mathbf{H}[\ell]
e^{-j2\pi q\ell/N_{\mathrm{fft}}},
\qquad
\mathbf{H}_f[q]\in\mathbb{C}^{N_r\times N_t}.
\]
After applying the FFT independently at each receive antenna,
\begin{equation}
\mathbf{y}_f[n,q]
=
\mathbf{H}_f[q]\mathbf{x}_f[n,q]
+
\mathbf{v}_f[n,q].\nonumber
\end{equation}

A conventional frequency-domain receiver then forms
\[
\widehat{\mathbf{x}}_f[n,q]
=
\mathbf{G}_f[n,q]\mathbf{y}_f[n,q],
\qquad
\mathbf{G}_f[n,q]\in\mathbb{C}^{N_t\times N_r}.
\]
For a MP channel, a stable causal inverse can be constructed under appropriate conditions, whereas NMP channel zeros outside the unit circle lead to unstable or noncausal components under direct inversion. In the noisy case, exact inversion is also undesirable since spectral nulls of \(\mathbf{H}(z)\) are inverted into large gains that amplify \(\mathbf{V}_t(z)\); regularized inverse filters are therefore typically used instead~\cite{sayed2003fundamentals}. As illustrated in Fig.~\ref{fig:mimo_detect}, the received MIMO-OFDM signal can be processed by the neural receiver operating directly in the time domain and a conventional LMMSE receiver based on frequency-domain MIMO-OFDM receive processing. Conventional frequency-domain OFDM receiver and neural receiver represent the same underlying inverse-filtering problem. Neural receiver can approximate the corresponding deconvolution operator through temporal filtering of the received samples, whereas frequency-domain receive processing is performed independently on each subcarrier. 

\subsection{Time-Domain Stable and Causal Equalization of Mixed-Phase Channels}
\label{subsec:iir_fir_nmp}
We first characterize the filtering structure of a general stable causal inverse of mixed-phase channels and then specialize to the LMMSE filter considered in this work. Let $G_{\mathrm{eq}}^{(\Delta)}(z)$ denote a causal channel inverse with decision delay $\Delta$. For a rational channel, any stable causal rational realization of $G_{\mathrm{eq}}^{(\Delta)}(z)$ can, under distinct poles, be expressed as
\begin{equation}
\label{eq:general_parallel_equalizer}
G_{\mathrm{eq}}^{(\Delta)}(z)
=
\underbrace{
\sum_{m=1}^{M} \frac{c_{m,\Delta}}{1-p_m z^{-1}}
}_{Q_{\Delta}(z)}
+
\underbrace{
\sum_{d=0}^{D_{\Delta}-1} f_{d,\Delta} z^{-d}
}_{P_{\Delta}(z^{-1})},
\qquad |p_m| < 1,
\end{equation}
where $Q_{\Delta}(z)$ and $P_{\Delta}(z^{-1})$ are the stable IIR and finite causal FIR components, respectively. Thus, the parallel IIR--FIR interpretation is not specific to a particular receiver design criterion; different receiver filters simply induce different poles, residues, and FIR coefficients. Accordingly, the recurrent and skip branches of the introduced skip-RNN are configured to approximate the dominant IIR and FIR components of the target filter.

In this work, we focus on the LMMSE receiver because it provides a noise-regularized channel inverse suitable for weight configuration while avoiding the severe noise enhancement of exact inversion along weak channel modes. The following development therefore shows how the causal LMMSE receiver acquires the general IIR--FIR form in \eqref{eq:general_parallel_equalizer}.

\begin{assumption}
\label{ass:statistics}
The transmitted symbols and additive noise are mutually independent, zero-mean white processes with powers $\sigma_x^2$ and $\sigma_v^2>0$, respectively.
\end{assumption}

\begin{definition}[Mixed-phase channel]
\label{def:mixed_phase_channel}
Let $H(z)$ be a causal and stable channel with no zeros on the unit circle. It admits the factorization
\begin{equation}
\label{eq:min_allpass_factorization}
H(z)=B(z)A(z),
\end{equation}
where $B(z)$ is causal, stable, and minimum phase, while $A(z)$ is a causal stable all-pass factor satisfying
\begin{equation}
\label{eq:allpass_property}
A(z)A^{\dagger}(z)=1.
\end{equation}
The channel is mixed phase when $A(z)$ is nontrivial, corresponding to one or more zeros of $H(z)$ outside the unit circle.
\end{definition}

The significance of this factorization follows from the exact channel inverse, $H^{-1}(z)=A^{\dagger}(z)B^{-1}(z).$
Since $B(z)$ is minimum phase, $B^{-1}(z)$ admits a causal and stable realization. In contrast, for a nontrivial all-pass factor, $A^{-1}(z)=A^{\dagger}(z)$ admits a stable anti-causal realization. Consequently, exact stable inversion of an NMP channel is generally noncausal, whereas enforcing causality on the exact inverse leads to instability~\cite{moir2007z,moir2013inverting}. A realizable receiver must therefore approximate or regularize this two-sided inverse structure.

For the receive processing considered in this work to configure the weights of skip-RNN, the unconstrained LMMSE filter for estimating the delayed symbol $x[\tau-\Delta]$ is~\cite{sayed2003fundamentals,kailath2000linear}
\begin{equation}
\label{eq:mmse_nmp_inverse}
W_{\mathrm{nc}}^{(\Delta)}(z)
=
\frac{
\sigma_x^2 z^{-\Delta}H^{\dagger}(z)
}{
\sigma_x^2H(z)H^{\dagger}(z)+\sigma_v^2
}.
\end{equation}
Equation~\eqref{eq:mmse_nmp_inverse} can be interpreted as a noise-regularized channel inverse. Unlike exact inversion, it limits amplification along weak channel modes. Nevertheless, the unconstrained solution is generally two-sided for a mixed-phase channel and therefore cannot be directly implemented as a causal receiver.

\begin{definition}[Causal--anti-causal decomposition]
\label{def:causal_projection}
For a stable two-sided Laurent expansion of $F(z)=\sum_{m=-\infty}^{\infty}f[m]z^{-m}$, define
\begin{equation}
[F(z)]_{+}
:=
\sum_{m\geq0}f[m]z^{-m},
\qquad
[F(z)]_{-}
:=
\sum_{m<0}f[m]z^{-m},\nonumber
\end{equation}
such that $F(z)=[F(z)]_{+}+[F(z)]_{-}$.
\end{definition}
The connection between the LMMSE inverse filter and exact channel inversion follows directly in the zero-noise limit.

\begin{lemma}[Noiseless limit]
\label{lem:noiseless_limit}
For the unconstrained LMMSE inverse filter in \eqref{eq:mmse_nmp_inverse},
\begin{equation}
\lim_{\sigma_v^2\rightarrow0}
W_{\mathrm{nc}}^{(\Delta)}(z)
=
z^{-\Delta}H^{-1}(z).\nonumber
\end{equation}
\end{lemma}
\begin{proof}
Taking $\sigma_v^2\rightarrow0$ in \eqref{eq:mmse_nmp_inverse} gives
\begin{align}
\lim_{\sigma_v^2\rightarrow0}
W_{\mathrm{nc}}^{(\Delta)}(z)
&=
\frac{
z^{-\Delta}H^{\dagger}(z)
}{
H(z)H^{\dagger}(z)
}
\nonumber\\
&=
z^{-\Delta}H^{-1}(z).\nonumber
\end{align}
\end{proof}

Using the mixed-phase factorization in \eqref{eq:min_allpass_factorization} and the all-pass property in \eqref{eq:allpass_property},
\begin{equation}
H(z)H^{\dagger}(z)
=
B(z)B^{\dagger}(z).\nonumber
\end{equation}
Hence, the received-signal power spectrum admits the spectral factorization
\begin{equation}
\label{eq:received_spectrum_scalar}
\sigma_x^2B(z)B^{\dagger}(z)+\sigma_v^2
=
\Phi(z)\Phi^{\dagger}(z),
\end{equation}
where $\Phi(z)$ is chosen as the canonical causal, stable, and minimum-phase spectral factor~\cite{sayed2003fundamentals,kailath2000linear}.
Substituting $H^{\dagger}(z)=B^{\dagger}(z)A^{\dagger}(z)$ and \eqref{eq:received_spectrum_scalar} into \eqref{eq:mmse_nmp_inverse} yields
\begin{equation}
W_{\mathrm{nc}}^{(\Delta)}(z)
=
\frac{1}{\Phi(z)}
z^{-\Delta}
\frac{
\sigma_x^2B^{\dagger}(z)A^{\dagger}(z)
}{
\Phi^{\dagger}(z)
}.
\nonumber
\end{equation}
The factor $1/\Phi(z)$ is causal and stable because $\Phi(z)$ is minimum phase, whereas the remaining factor contains the stable left-sided structure associated with the mixed-phase channel inverse. 
To obtain a realizable causal receiver filter, the Wiener--Hopf construction projects the two-sided solution onto its causal part. The optimal causal LMMSE receiver filter with decision delay $\Delta$ is therefore~\cite{sayed2003fundamentals,kailath2000linear}
\begin{equation}
\label{eq:causal_mmse_wiener_hopf}
W_{\mathrm{c}}^{(\Delta)}(z)
=
\frac{1}{\Phi(z)}
\left[
z^{-\Delta}
\underbrace{
\frac{
\sigma_x^2B^{\dagger}(z)A^{\dagger}(z)
}{
\Phi^{\dagger}(z)
}
}_{R(z)}
\right]_{+}.
\end{equation}
Equation~\eqref{eq:causal_mmse_wiener_hopf} makes the role of the decision delay explicit. Since $R(z)$ admits a stable left-sided realization, write
\begin{equation}
R(z)
=
\sum_{k=0}^{\infty}
r[-k]z^{k}.
\nonumber
\end{equation}
Introducing the decision delay gives
\begin{equation}
z^{-\Delta}R(z)
=
\sum_{k=0}^{\infty}
r[-k]z^{-(\Delta-k)}.
\nonumber
\end{equation}
The terms with $k\leq\Delta$ become causal, whereas those with $k>\Delta$ remain anti-causal and are removed by the causal projection. Therefore,
\begin{equation}
\label{eq:finite_delay_projection}
\left[
z^{-\Delta}R(z)
\right]_{+}
=
\sum_{k=0}^{\Delta}
r[-k]z^{-(\Delta-k)}.
\end{equation}
A finite decision delay thus retains the first $\Delta+1$ coefficients of the stable left-sided component as a causal FIR polynomial. Increasing $\Delta$ preserves a larger portion of the unconstrained LMMSE solution at the cost of increased detection latency.
Substituting \eqref{eq:finite_delay_projection} into \eqref{eq:causal_mmse_wiener_hopf} yields
\begin{equation}
\label{eq:causal_mmse_iir_fir_cascade}
W_{\mathrm{c}}^{(\Delta)}(z)
=
\frac{1}{\Phi(z)}
\underbrace{
\sum_{k=0}^{\Delta}
r[-k]z^{-(\Delta-k)}
}_{F_{\Delta}(z)}.
\end{equation}
Hence, the causal LMMSE inverse filter is a stable causal rational transfer function obtained as the cascade of the stable rational factor $1/\Phi(z)$ and the finite causal polynomial $F_{\Delta}(z)$. 

\begin{proposition}[Finite-delay IIR--FIR representation of the causal LMMSE inverse filter]
\label{prop:mmse_iir_fir_structure}
Assume that $\Phi(z)$ in \eqref{eq:received_spectrum_scalar} is rational and that the poles of $1/\Phi(z)$ are distinct. Then, for any finite decision delay $\Delta\geq0$ and for $D_{\Delta}$ FIR taps, the causal LMMSE filter in \eqref{eq:causal_mmse_iir_fir_cascade} admits the representation
\begin{equation}
\label{eq:unified_parallel_form}
W_{\mathrm{c}}^{(\Delta)}(z)
=
\underbrace{
\sum_{m=1}^{M}
\frac{c_{m,\Delta}}
{1-p_mz^{-1}}
}_{Q_{\Delta}(z)}
+
\underbrace{
\sum_{d=0}^{D_{\Delta}-1}
f_{d,\Delta}z^{-d}
}_{P_{\Delta}(z^{-1})},
\;
|p_m|<1.
\end{equation}
The first term contains the stable IIR modes of the causal LMMSE filter, while the second term is a finite causal FIR component from the rational realization of the polynomial
\end{proposition}

\begin{proof}
See Appendix~\ref{append_prop_mmse_iir_fir_structure}.
\end{proof}

\begin{remark}[Extension to MIMO LMMSE filtering]
\label{rem:mimo_extension}
The preceding derivation is presented for a scalar channel to expose the mixed-phase structure. For a MIMO channel $\mathbf H(z)\in\mathbb C^{N_r\times N_t}$,
the unconstrained LMMSE filter is obtained from the channel and received-signal spectral-density matrices, while its causal realization follows from matrix Wiener--Hopf spectral factorization~\cite{kailath2000linear,sayed2003fundamentals}. Non-minimum-phase transmission zeros induce a corresponding two-sided inverse structure. With a finite decision delay, causal projection retains a finite portion of this structure, yielding a stable causal rational matrix filter.

For rational matrix transfer functions with simple poles, the resulting causal LMMSE inverse admits the matrix-valued parallel representation
\begin{equation}
\mathbf W_{\mathrm c}^{(\Delta)}(z)
=
\sum_{m=1}^{M}
\frac{\mathbf C_{m,\Delta}}
{1-p_mz^{-1}}
+
\sum_{d=0}^{D_{\Delta}-1}
\mathbf F_{d,\Delta}z^{-d},\nonumber
\end{equation}
where $\mathbf C_{m,\Delta}\in\mathbb C^{N_t\times N_r}$ and $\mathbf F_{d,\Delta}\in\mathbb C^{N_t\times N_r}$ are the matrix-valued residues and FIR coefficients, respectively. This matrix IIR--FIR structure provides the target representation used by the skip-RNN receiver.
\end{remark}


\section{Online Real-Time Neural Receiver for Mixed-Phase Channels}
\label{sec:online_skip_rnn}

This section develops the introduced online \emph{skip-RNN} receiver and shows, step by step, how its recurrent and skip branches realize the mixed IIR--FIR filtering structure established in Proposition~\ref{prop:mmse_iir_fir_structure} for stable causal processing over mixed-phase channels.

\begin{figure}[!ht]
    \centering
    \includegraphics[width=\linewidth]{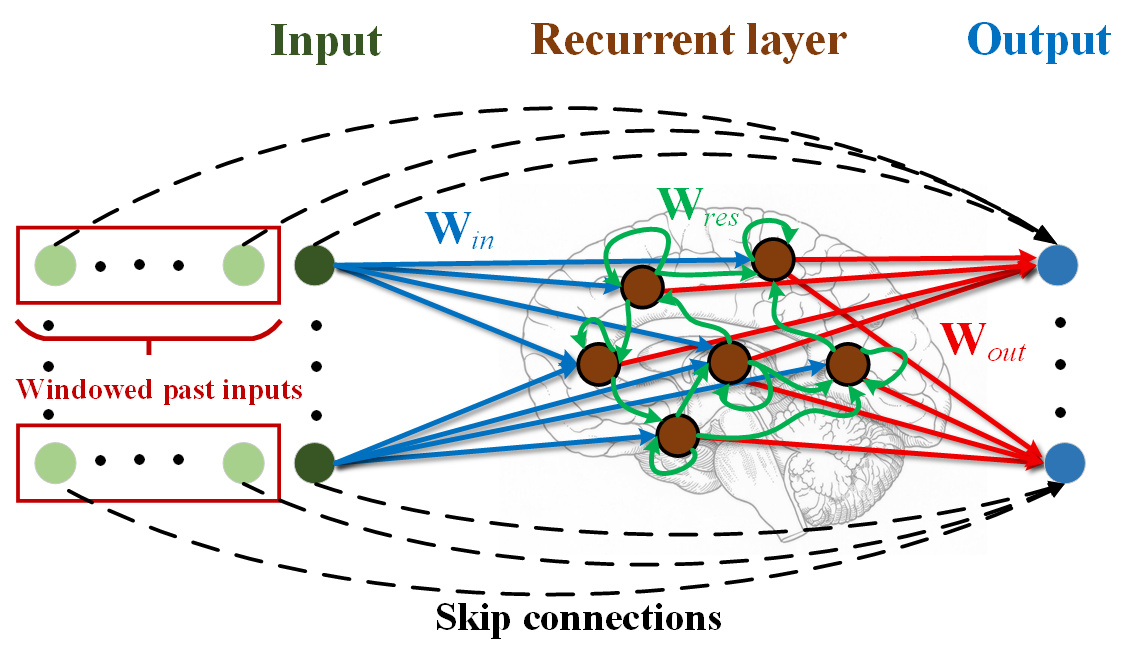}
    \caption{Skip-RNN Neural Receiver Architecture}
    \label{fig:skip-RNN}
\end{figure}

\subsection{Skip-Augmented Recurrent Receiver}
\label{subsec:skip_rnn_architecture}

This section describes the architecture of the skip-RNN neural receiver shown in Fig.~\ref{fig:skip-RNN}. Let $\tau\in\{0,1,\ldots,N-1\}$ index the time-domain samples within TTI $t$. The received sample is $\mathbf y_t[\tau]\in\mathbb C^{N_r}$ and the transmitted sample is $\mathbf x_t[\tau]\in\mathbb C^{N_t}$.

\paragraph{Recurrent branch.}
Skip-RNN consists of a linear recurrent layer with $N_n$ units maintains the state $\mathbf s_{\mathrm{res}}[\tau]\in\mathbb C^{N_n}$ according to
\begin{equation}
\label{eq:rnn_state_update}
\mathbf s_{\mathrm{res}}[\tau]=
\mathbf W_{\mathrm{res}}\mathbf s_{\mathrm{res}}[\tau-1]+
\mathbf W_{\mathrm{in}}\mathbf y_t[\tau],
\end{equation}
where $\mathbf W_{\mathrm{in}}\in\mathbb C^{N_n\times N_r}$, $\mathbf W_{\mathrm{res}}\in\mathbb C^{N_n\times N_n}$.

\paragraph{Windowed input.}
In parallel with the recurrent state, the current received sample and its $N_w-1$ predecessors are stacked into the \emph{skip window}
\begin{equation}
\label{eq:skip_window}
\mathbf y_{t,w}[\tau]:=\big[
\mathbf y_t^{\mathsf T}[\tau],\;
\mathbf y_t^{\mathsf T}[\tau-1],\;
\ldots,\;
\mathbf y_t^{\mathsf T}[\tau-N_w+1]
\big]^{\mathsf T}
\in\mathbb C^{N_rN_w}.
\end{equation}

\paragraph{Readout.}
The recurrent and skip features are concatenated as
\begin{equation}
\mathbf z_t[\tau]
=
\begin{bmatrix}
\mathbf s_{\mathrm{res}}[\tau]\\[2pt]
\mathbf y_{t,w}[\tau]
\end{bmatrix}
\in\mathbb C^{d},
\qquad
d:=N_n+N_rN_w,\nonumber
\end{equation}
and the detected time-domain sample is obtained through the linear readout
\begin{equation}
\widehat{\mathbf x}_t[\tau]=
\mathbf W_{\mathrm{out}}\mathbf z_t[\tau],
\qquad
\mathbf W_{\mathrm{out}}\in\mathbb C^{N_t\times d}.\nonumber
\end{equation}
The matrices $\mathbf W_{\mathrm{in}}$ and $\mathbf W_{\mathrm{res}}$ are fixed before online operation, whereas $\mathbf W_{\mathrm{out}}$ is adapted online from the available pilot observations. Setting $N_w=0$ removes the skip branch and recovers the conventional recurrent receiver.

\subsection{Time-Domain Deconvolution Realized by the Skip-RNN}
\label{subsec:circular_realization}

In this section we establish the direct time-domain input--output relation of the receiver. For this purpose, we consider the following assumption.


\begin{assumption}[Reservoir]
\label{as:reservoir}
The recurrent matrix $\mathbf W_{\mathrm{res}}$ is diagonalizable and satisfies
$\rho_{\mathrm{res}}:=\rho(\mathbf W_{\mathrm{res}})<1$. Its eigenvalues
$\lambda_1,\ldots,\lambda_{N_n}$ are assumed distinct and nonzero.
\end{assumption}

Repeated substitution in~\eqref{eq:rnn_state_update} gives
\begin{align}
\mathbf s_{\mathrm{res}}[\tau]
&=
\mathbf W_{\mathrm{res}}^{K+1}
\mathbf s_{\mathrm{res}}[\tau-K-1]
+
\sum_{k=0}^{K}
\mathbf W_{\mathrm{res}}^k
\mathbf W_{\mathrm{in}}
\mathbf y_t[\tau-k].\nonumber
\end{align}

Since $\rho(\mathbf W_{\mathrm{res}})<1$ under Assumption~\ref{as:reservoir},
\begin{equation}
\mathbf W_{\mathrm{res}}^{K+1}\rightarrow\mathbf 0,
\qquad K\rightarrow\infty.\nonumber
\end{equation}
Therefore, after the initial-state transient has decayed, the recurrent state can be written as
\begin{equation}
\label{eq:reservoir_state_convolution}
\mathbf s_{\mathrm{res}}[\tau]=
\sum_{k=0}^{\infty}
\mathbf W_{\mathrm{res}}^k
\mathbf W_{\mathrm{in}}
\mathbf y_t[\tau-k].
\end{equation}

Partition the readout matrix according to the recurrent and skip features as
\begin{equation}
\mathbf W_{\mathrm{out}}=
\big[\mathbf W_{\mathrm{out,s}}
\;\;
\mathbf W_{\mathrm{out,w}}\big],\nonumber
\end{equation}
where
\begin{equation}
\mathbf W_{\mathrm{out,w}}=
\big[\mathbf W_{\mathrm{F},0}
\;\cdots\;
\mathbf W_{\mathrm{F},N_w-1}\big]. \nonumber
\end{equation}
Substituting \eqref{eq:reservoir_state_convolution} and \eqref{eq:skip_window} into the readout gives
\begin{align}
\widehat{\mathbf x}_t[\tau]
&=
\sum_{k=0}^{\infty}
\mathbf W_{\mathrm{out,s}}
\mathbf W_{\mathrm{res}}^k
\mathbf W_{\mathrm{in}}
\mathbf y_t[\tau-k]
\nonumber\\
&\quad+
\sum_{\ell=0}^{N_w-1}
\mathbf W_{\mathrm{F},\ell}
\mathbf y_t[\tau-\ell].
\label{eq:skip_rnn_time_branches}
\end{align}

Hence, both branches operate directly on delayed received samples. Define the equivalent impulse-response matrix
\begin{equation}
\mathbf G_{\mathrm{SR}}[k]=
\mathbf W_{\mathrm{out,s}}
\mathbf W_{\mathrm{res}}^k
\mathbf W_{\mathrm{in}}+
\begin{cases}
\mathbf W_{\mathrm{F},k}, & 0\leq k<N_w,
\\
\mathbf 0, & k\geq N_w.\nonumber
\end{cases}
\end{equation}
Then \eqref{eq:skip_rnn_time_branches} becomes
\begin{equation}
\label{eq:skip_rnn_direct_convolution}
\widehat{\mathbf x}_t[\tau]
=
\sum_{k=0}^{\infty}
\mathbf G_{\mathrm{SR}}[k]
\mathbf y_t[\tau-k]
=
\left(
\mathbf G_{\mathrm{SR}}
*
\mathbf y_t
\right)[\tau].
\end{equation}

Equation~\eqref{eq:skip_rnn_direct_convolution} establishes that, the skip-RNN is equivalent to a causal MIMO deconvolution filter. The recurrent branch contributes an infinite-duration impulse response, whereas the skip branch contributes a finite-duration response over the first $N_w$ delays. We next move to the $z$-domain to expose the corresponding IIR and FIR basis functions explicitly.

\subsection{IIR--FIR Filter Representation in the $z$-Domain}
\label{subsec:basis_functions}

Having established the time-domain deconvolution implemented by the skip-RNN, we now express the same input--output relation in the $z$-domain. This representation makes the pole structure of the recurrent branch and the finite-delay structure of the skip branch explicit.

\subsubsection{Recurrent IIR basis}

Taking the one-sided $z$-transform of the linear recurrent update with zero initial state gives
\begin{equation}
\mathbf S_{\mathrm{res}}(z)
=
\mathbf W_{\mathrm{res}}z^{-1}
\mathbf S_{\mathrm{res}}(z)
+
\mathbf W_{\mathrm{in}}
\mathbf Y_t(z).\nonumber
\end{equation}
Since $\mathbf I-\mathbf W_{\mathrm{res}}z^{-1}$ is invertible for
$|z|>\rho_{\mathrm{res}}$,
\begin{equation}
\label{eq:state_ztransform}
\mathbf S_{\mathrm{res}}(z)=
\left(\mathbf I-\mathbf W_{\mathrm{res}}z^{-1}
\right)^{-1}
\mathbf W_{\mathrm{in}}\mathbf Y_t(z).
\end{equation}

Equation~\eqref{eq:state_ztransform} is the $z$-domain counterpart of the infinite convolution in \eqref{eq:reservoir_state_convolution}. To expose the individual pole basis functions, use the eigendecomposition
\begin{equation}
\mathbf W_{\mathrm{res}}=
\mathbf V\boldsymbol\Lambda\mathbf V^{-1},
\qquad
\boldsymbol\Lambda=
\operatorname{diag}
(\lambda_1,\ldots,\lambda_{N_n}).\nonumber
\end{equation}
Let $\mathbf v_j$ denote the $j$th column of $\mathbf V$ and
$\mathbf u_j^{\mathsf H}$ the $j$th row of $\mathbf V^{-1}$. Then
\begin{equation}
\left(
\mathbf I-\mathbf W_{\mathrm{res}}z^{-1}
\right)^{-1}=
\sum_{j=1}^{N_n}
\frac{
\mathbf v_j\mathbf u_j^{\mathsf H}
}{1-\lambda_jz^{-1}}. \nonumber
\end{equation}

Define the modal state and modal input matrix as
\begin{equation}
\widetilde{\mathbf S}_{\mathrm{res}}(z):=
\mathbf V^{-1}\mathbf S_{\mathrm{res}}(z),
\qquad
\widetilde{\mathbf W}_{\mathrm{in}}:=
\mathbf V^{-1}\mathbf W_{\mathrm{in}}.\nonumber
\end{equation}
Then~\eqref{eq:state_ztransform} becomes
\begin{equation}
\widetilde{\mathbf S}_{\mathrm{res}}(z)=
\underbrace{\left(
\mathbf I-\boldsymbol\Lambda z^{-1}
\right)^{-1}
\widetilde{\mathbf W}_{\mathrm{in}}
}_{\boldsymbol\Omega_{\mathrm{IIR}}(z)}
\mathbf Y_t(z),\nonumber
\end{equation}
where
\begin{equation}
\boldsymbol\Omega_{\mathrm{IIR}}(z)=
\begin{bmatrix}
\dfrac{
\widetilde{\mathbf w}_{\mathrm{in}}^{(1)\mathsf T}
}{
1-\lambda_1z^{-1}
}\\
\vdots\\
\dfrac{\widetilde{\mathbf w}_{\mathrm{in}}^{(N_n)\mathsf T}
}{1-\lambda_{N_n}z^{-1}} \nonumber
\end{bmatrix},
\end{equation}
and $\widetilde{\mathbf w}_{\mathrm{in}}^{(j)\mathsf T}$ denotes the $j$th row of $\widetilde{\mathbf W}_{\mathrm{in}}$.

Thus, in modal coordinates, the recurrent branch consists of $N_n$ independent first-order IIR modes with poles
$\{\lambda_j\}_{j=1}^{N_n}$.

\begin{figure}[!t]
\centering
\includegraphics[width=\linewidth]{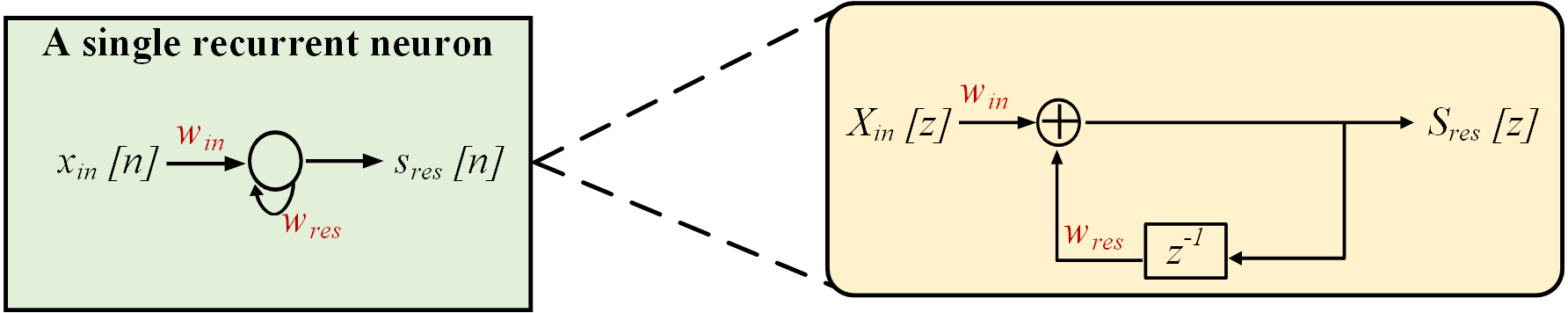}
\caption{Signal-flow representation of a linear recurrent mode interpreted as a first-order IIR filter.}
\label{fig:neuron_iir}
\end{figure}

The reservoir contribution at the receiver output can be expressed in the same coordinates by defining $\widetilde{\mathbf W}_{\mathrm{out,s}}:=\mathbf W_{\mathrm{out,s}}\mathbf V$.
Therefore, $\mathbf W_{\mathrm{out,s}} \mathbf S_{\mathrm{res}}(z) = \widetilde{\mathbf W}_{\mathrm{out,s}} \widetilde{\mathbf S}_{\mathrm{res}}(z)$.
The modal transformation therefore leaves the input--output transfer function unchanged and only reveals the recurrent dynamics as a bank of first-order IIR basis functions.

\subsubsection{Skip FIR basis}

Applying
$\mathcal Z\{\mathbf y_t[\tau-\ell]\}
=
z^{-\ell}\mathbf Y_t(z)$
to the windowed input in \eqref{eq:skip_window} gives
\begin{equation}
\mathbf Y_{t,w}(z)=
\boldsymbol\Omega_{\mathrm{FIR}}(z)
\mathbf Y_t(z),\nonumber
\end{equation}
where
\begin{equation}
\boldsymbol\Omega_{\mathrm{FIR}}(z)=
\begin{bmatrix}
\mathbf I_{N_r}\\
z^{-1}\mathbf I_{N_r}\\\vdots\\
z^{-(N_w-1)}\mathbf I_{N_r}
\end{bmatrix}\in
\mathbb C^{N_rN_w\times N_r}.\nonumber
\end{equation}
Thus, the skip branch provides the complete causal FIR delay basis over the first $N_w$ samples.

\subsubsection{Combined IIR--FIR representation}

Stacking the recurrent and skip features gives
\begin{equation}
\boldsymbol\Omega_{\mathrm{SR}}(z)=
\begin{bmatrix}
\boldsymbol\Omega_{\mathrm{IIR}}(z)\\
\boldsymbol\Omega_{\mathrm{FIR}}(z)
\end{bmatrix},\nonumber
\end{equation}
such that
\begin{equation}
\mathbf Z_t(z)=
\boldsymbol\Omega_{\mathrm{SR}}(z)
\mathbf Y_t(z).\nonumber
\end{equation}

Using the transformed recurrent readout, the exact input--output relation of the skip-RNN is
\begin{equation}
\widehat{\mathbf X}_t(z)
=
\mathbf G_{\mathrm{SR}}(z)
\mathbf Y_t(z),
\end{equation}
where
\begin{equation}
\label{eq:skip_rnn_mixed_iir_fir_transfer}
\mathbf G_{\mathrm{SR}}(z)
=
\underbrace{
\widetilde{\mathbf W}_{\mathrm{out,s}}
\boldsymbol\Omega_{\mathrm{IIR}}(z)
}_{\text{recurrent IIR branch }(\widehat{\mathbf Q}_{N_n}(z))}
+
\underbrace{
\sum_{\ell=0}^{N_w-1}
\mathbf W_{\mathrm{F},\ell}z^{-\ell}
}_{\text{skip FIR branch }(\widehat{\mathbf F}_{N_w}(z))}.
\end{equation}


The poles $\{\lambda_j\}$ are determined by the fixed recurrent matrix $\mathbf W_{\mathrm{res}}$, while the effective residues and FIR coefficients are determined through the readout weights together with the fixed input mapping. Consequently, the recurrent branch spans a structured IIR function space, whereas the skip branch spans a finite-dimensional FIR delay space.

Comparing \eqref{eq:skip_rnn_mixed_iir_fir_transfer} with \eqref{eq:unified_parallel_form} shows that the skip-RNN realizes the parallel IIR--FIR structure of the causal and stable inversion for mixed-phase channels established in Proposition~\ref{prop:mmse_iir_fir_structure}. This correspondence motivates configuring the fixed recurrent and skip parameters so that the resulting basis functions are well matched to the dominant structure of the target function.

\subsection{Geometric Interpretation and Weight Configuration Problem Formulation}
\label{subsec:detection_mse_lmmse_projection}

As shown in Section~\ref{subsec:basis_functions}, the recurrent branch provides a set of stable IIR basis functions determined by $\mathbf W_{\mathrm{res}}$ and $\mathbf W_{\mathrm{in}}$, while the skip branch provides a finite-delay FIR basis determined by the window length $N_w$. Define
\begin{equation}
\boldsymbol\Theta_{\mathrm{IIR}}
=
\left\{
\mathbf W_{\mathrm{res}},
\mathbf W_{\mathrm{in}}
\right\},
\qquad
\boldsymbol\Theta=
\left\{
\boldsymbol\Theta_{\mathrm{IIR}},
N_w
\right\}.
\end{equation}
For a fixed configuration $\boldsymbol\Theta$, the skip-RNN spans the subspace
\begin{equation}
\Omega_{\mathrm{SR}}(\boldsymbol\Theta)
=
\Omega_{\mathrm{IIR}}(\boldsymbol\Theta_{\mathrm{IIR}})
+
\Omega_{\mathrm{FIR}}(N_w).\nonumber
\end{equation}
Thus, $\mathbf W_{\mathrm{res}}$ and $\mathbf W_{\mathrm{in}}$ determine the available recurrent IIR basis, whereas $N_w$ determines the dimension of the FIR basis.

For a given received block, the output weights are trained from the available reference symbols according to
\begin{equation}
\label{eq:finite_pilot_readout_training}
\widehat{\mathbf W}_{\mathrm{out}}=
\operatorname*{arg\,min}_{\mathbf W_{\mathrm{out}}}
\frac{1}{T}\left\|
\mathbf W_{\mathrm{out}}
\mathbf Z_t(\boldsymbol\Theta)-
\mathbf X_t\right\|_F^2.
\end{equation}
For fixed $\boldsymbol\Theta$, the feature matrix $\mathbf Z_t(\boldsymbol\Theta)$ is fixed, so \eqref{eq:finite_pilot_readout_training} is a standard LS problem. Training $\mathbf W_{\mathrm{out}}$ changes only the linear combination of the available IIR and FIR basis functions; it does not modify the basis functions themselves.

\begin{figure}[!ht]
\centering
\includegraphics[width=\linewidth]{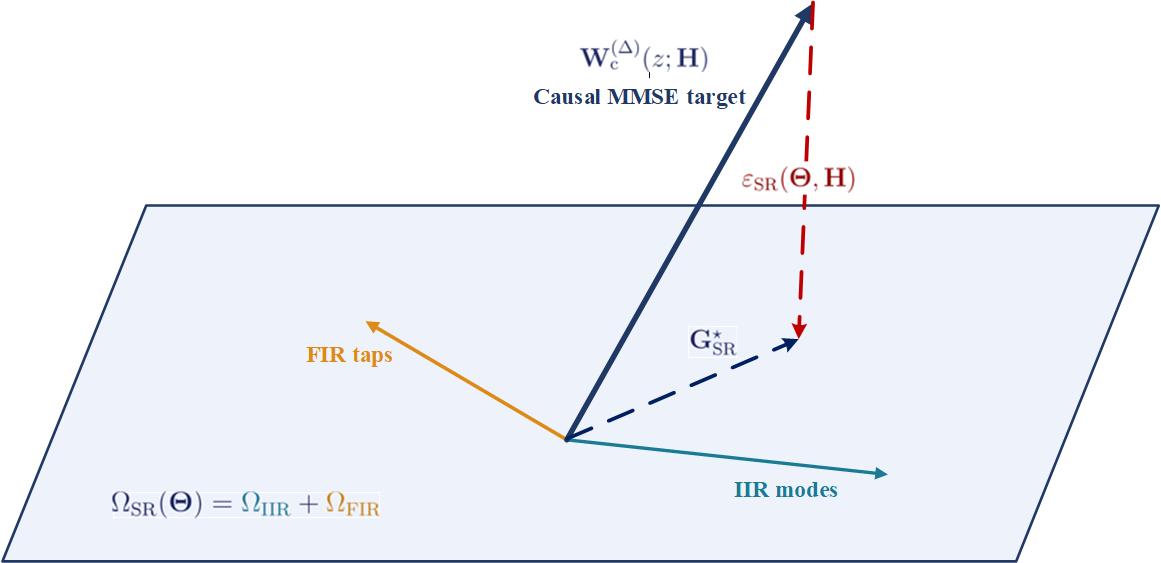}
\caption{Geometric interpretation of skip-RNN receiver. The recurrent and skip features span the joint IIR-FIR space $\Omega_{\mathrm{SR}}$, onto which the LMMSE inverse filter is projected through the readout.}
\label{fig:skip_rnn_geometric}
\end{figure}

\paragraph{Geometric Interpretation.}
This is precisely the situation described by the geometric interpretation of reservoir-based receiver introduced in our previous works~\cite{10974467,jere2023channel}: the fixed units generate a set of basis functions, and training the output weights selects the linear combination of them that best approximates a desired target function $f(\cdot)$. Under an $\ell_2$ loss that combination is the orthogonal projection of $f(\cdot)$ onto
the spanned subspace,
\begin{equation}
f_\perp(\cdot)=\operatorname{proj}_{\Omega_{\mathrm{SR}}(\boldsymbol\Theta)}\{f(\cdot)\},\nonumber
\end{equation}

taken here under the metric induced by the received signal over the pilot block. Fig.~\ref{fig:skip_rnn_geometric} illustrates this for the skip-RNN.

\paragraph{Representation error induced by the fixed basis.}
For a fixed configuration $\boldsymbol\Theta$, every choice of $\mathbf W_{\mathrm{out}}$ produces a receiver response belonging to $\Omega_{\mathrm{SR}}(\boldsymbol\Theta)$. Therefore, even with an optimally chosen readout, any component of the target filter lying outside this space cannot be represented by the skip-RNN. The resulting representation error is
\begin{equation}
\varepsilon_{\mathrm{SR}}^2(\boldsymbol\Theta)
= \left\| f(\cdot) - \operatorname{proj}_{\Omega_{\mathrm{SR}}(\boldsymbol\Theta)}
\{f(\cdot)\} \right\|_2^2 . \nonumber
\end{equation}

This quantity is the \emph{representation error}: it measures how closely the function space generated by the skip-RNN's fixed basis can represent the target function.

\paragraph{Why random basis functions are not enough}
basis mismatch and finite-sample readout estimation constitute distinct sources of error. This distinction explains why random reservoir initialization can be inefficient. Random $\mathbf W_{\mathrm{res}}$ places the poles, and hence $\Omega_{\mathrm{IIR}} (\boldsymbol\Theta)$, without any prior knowledge to the channel inverses the receiver will actually approximate. The resulting $\varepsilon_{\mathrm{SR}}$ is therefore uncontrolled: it varies from draw to draw and from channel to channel, and nothing in \eqref{eq:finite_pilot_readout_training} can compensate for a poor draw, since training cannot leave the subspace. Configuring $\boldsymbol\Theta$ from the statistical knowledge of the target is what removes this dependence on chance.

\paragraph{Causal LMMSE filter as the target function}
For a channel realization $\mathbf H$, the target transfer function is the stable causal LMMSE filter with decision delay $\Delta$, denoted by $\mathbf W_{\mathrm c}^{(\Delta)}(z;\mathbf H)$ and derived in Sec.~\ref{subsec:iir_fir_nmp}. The skip-RNN can realize only transfer functions belonging to the space $\Omega_{\mathrm{SR}}(\boldsymbol\Theta)$ determined by its fixed configuration. Hence, the best realizable receiver response is the projection of the causal LMMSE target onto this space,
\begin{equation}
\mathbf G_{\mathrm{SR}}^{\star}
(z;\boldsymbol\Theta,\mathbf H)=
\operatorname{proj}_{\Omega_{\mathrm{SR}}(\boldsymbol\Theta)}
\left\{\mathbf W_{\mathrm c}^{(\Delta)}(z;\mathbf H)\right\}. \nonumber
\end{equation}

The corresponding representation error is
\begin{equation}
\varepsilon_{\mathrm{SR}}^2
(\boldsymbol\Theta,\mathbf H)=
\left\|
\mathbf W_{\mathrm c}^{(\Delta)}(z;\mathbf H)-
\mathbf G_{\mathrm{SR}}^{\star}
(z;\boldsymbol\Theta,\mathbf H)\right\|_2^{2}. \nonumber
\end{equation}

Thus, for a fixed configuration $\boldsymbol\Theta$, readout training only selects a linear combination of the available IIR and FIR basis functions and cannot change the space $\Omega_{\mathrm{SR}}(\boldsymbol\Theta)$ itself. The remaining approximation error is therefore determined by how well the fixed skip-RNN basis represents $\mathbf W_{\mathrm c}^{(\Delta)}(z;\mathbf H)$. Consequently, the two stages play distinct roles. Online readout training adapts the receiver response within $\Omega_{\mathrm{SR}}(\boldsymbol\Theta)$, whereas the configuration of $\mathbf W_{\mathrm{res}}$, $\mathbf W_{\mathrm{in}}$, and $N_w$ determines the function space available for this adaptation. Weight configuration therefore aims to construct a basis whose span provides an accurate approximation of the target causal LMMSE filter.

\paragraph{Weight-configuration problem.}
The configuration $\boldsymbol\Theta$ is determined offline and must support different channel realizations. The appropriate design objective is therefore the expected representation error
\begin{equation}
\label{eq:problem_p3_projection}
\boldsymbol\Theta^{\star}=
\operatorname*{arg\,min}_{\boldsymbol\Theta\in\mathcal C}
\;
\mathbb E_{\mathbf H}\left[\varepsilon_{\mathrm{SR}}^2(\boldsymbol\Theta,\mathbf H)\right],
\end{equation}
where $\mathcal C$ denotes the admissible configuration set, including the reservoir stability constraint
$\rho(\mathbf W_{\mathrm{res}})<1$ and the prescribed receiver-complexity constraints on the recurrent dimension and skip-window length.

Problem~\eqref{eq:problem_p3_projection} shows that weight configuration is fundamentally a function-space design problem: the fixed IIR and FIR basis functions should be selected so that their joint span is well aligned with the ensemble of causal LMMSE inverse filters. In practice, the true distribution $p(\mathbf H)$ is unavailable. Instead, the corresponding causal LMMSE inverse of estimated channel realizations are used to infer the dominant structure of the target function space.

\section{Channel-Informed Weight Configuration of skip-RNN}
\label{sec:weight_configuration}

This section exploits the IIR--FIR structure established in Section~\ref{subsec:iir_fir_nmp} to construct this configuration. The dominant low-dimensional structure of the IIR ensemble is extracted and approximated using stable rational functions to configure $\mathbf W_{\mathrm{res}}$ and $\mathbf W_{\mathrm{in}}$, while the FIR structure determines an appropriate skip-window length. In addition, a statistical FIR readout prior is obtained from the channel ensemble and incorporated through prior-centered regularized online adaptation. Importantly, this prior does not alter $\Omega_{\mathrm{SR}}(\boldsymbol\Theta)$; rather, it improves estimation of the readout coefficients when only a limited number of pilot samples are available.

\subsection{Low-Dimensional Approximation of the IIR Component}
\label{subsec:low_dim_iir}

Proposition~\ref{prop:mmse_iir_fir_structure} shows that, for each channel realization $r$, the causal LMMSE filter can be written as
\begin{equation}
\label{eq:causal_mmse_realization_decomp}
W_{\mathrm c}^{(\Delta,r)}(z)
=
Q_{\Delta}^{(r)}(z)
+
P_{\Delta}^{(r)}(z^{-1}),
\end{equation}
where $Q_{\Delta}^{(r)}(z)$ contains the stable IIR modes and
$P_{\Delta}^{(r)}(z^{-1})$ is a finite FIR component. This decomposition matches the two branches of the skip-RNN. The skip branch already provides the delay basis
$\{1,z^{-1},\ldots,z^{-(N_w-1)}\}$, whereas the recurrent branch contains only $N_n$ fixed first-order IIR basis functions. Therefore, the main configuration problem is to construct a finite recurrent basis that captures the dominant structure of the realization-dependent IIR component $Q_{\Delta}^{(r)}(z)$.

To characterize the ensemble of IIR targets, the digital frequency interval is sampled at
\begin{equation}
\omega_k=
\frac{2\pi k}{N},
\qquad
k=0,1,\ldots,N-1.\nonumber
\end{equation}
For the $r$th channel realization, define
\begin{equation}
\mathbf q^{(r)}=
\begin{bmatrix}
Q_{\Delta}^{(r)}(e^{j\omega_0}) &
Q_{\Delta}^{(r)}(e^{j\omega_1}) &
\cdots &
Q_{\Delta}^{(r)}(e^{j\omega_{N-1}})
\end{bmatrix}^{T}.\nonumber
\end{equation}
For the MIMO case, the entries of the matrix-valued IIR filter at each frequency are vectorized and stacked.

Let denote an $M$-dimensional orthonormal basis,
\begin{equation}
\mathbf F_M =
\begin{bmatrix}
\mathbf f_1 &
\mathbf f_2 &
\cdots &
\mathbf f_M
\end{bmatrix},
\qquad
\mathbf F_M^H\mathbf F_M = \mathbf I_M, \nonumber
\end{equation}

The projection of $\mathbf q^{(r)}$ onto this subspace is
\begin{equation}
\widetilde{\mathbf q}_{M}^{(r)}=
\mathbf F_M\mathbf F_M^H
\mathbf q^{(r)}.\nonumber
\end{equation}

Define the second-moment matrix of the IIR target $\mathbf R_q$ estimated from $N_{\mathrm{obs}}$ channel realizations and its eigen-decomposition as
\begin{equation}
\widehat{\mathbf R}_q
=
\frac{1}{N_{\mathrm{obs}}}
\sum_{r=1}^{N_{\mathrm{obs}}}
\mathbf q^{(r)}
\mathbf q^{(r)H} = \mathbf Q \boldsymbol{\Lambda} \mathbf Q^H, \nonumber
\end{equation}
and the $M$ dominant eigenvectors of $\widehat{\mathbf R}_q$ are retained.

The optimum $M$-dimensional subspace is obtained from $\mathbf F_M^{\star} = \mathbf Q_M$, where $\mathbf Q_M$ contains the $M$ eigenvectors corresponding to the largest eigenvalues. The associated empirical projection error is
\begin{equation}
\varepsilon_{\mathrm{PCA}}^2(M)=
\frac{1}{N_{\mathrm{obs}}}
\sum_{r=1}^{N_{\mathrm{obs}}}
\left\|
\mathbf q^{(r)}-
\mathbf Q_M\mathbf Q_M^H\mathbf q^{(r)}
\right\|_2^2=
\sum_{m=M+1}^{N}\lambda_m .
\nonumber
\end{equation}

The dominant eigenvectors characterize the statistically significant IIR responses, but they are not themselves constrained to the first-order IIR form generated by the recurrent neurons. Therefore, each eigen mode is further approximated by a $K$th-order rational polynomial function~\cite{gustavsen1999rational}
\begin{equation}
f_m(z) \approx \widehat f_m(z) =
\sum_{k=1}^{K}
\frac{a_{m,k}}
{1-p_{m,k}z^{-1}},
\qquad
|p_{m,k}|<1, \nonumber
\end{equation}
where $f_m(z)$ denotes the transfer function associated with the $m$-th dominant eigenvector. Retaining $M$ dominant modes and assigning $K$ first-order IIR terms to each mode yields $N_n=MK$ configured recurrent weights. Their poles and residues determine the fixed recurrent configuration through $\mathbf W_{\mathrm{res}}$ and the input weight $\mathbf W_{\mathrm{in}}$.

\subsection{Approximation Error Bound}
\label{subsec:approximation_error_bound}

In practical, the weight configuration is obtained from $N_{\mathrm{obs}}$ estimated channel realizations. We therefore characterize the approximation capability of the configured skip-RNN using the corresponding empirical mean-squared $L_2$ error.

For channel realization $r$, let $Q_{\Delta}^{(r)}(z)$ denote the IIR component of the causal LMMSE filter and $\widetilde Q_{\Delta,M}^{(r)}(z)$ its reconstruction using the $M$ dominant PCA modes. The empirical PCA truncation error is
\begin{equation}
\label{eq:eps_pca}
\widehat{\varepsilon}_{\mathrm{PCA}}^{2}(M)
:=
\frac{1}{N_{\mathrm{obs}}}
\sum_{r=1}^{N_{\mathrm{obs}}}
\left\|
Q_{\Delta}^{(r)}(z)
-
\widetilde Q_{\Delta,M}^{(r)}(z)
\right\|_2^2.
\end{equation}
When the PCA is performed using the empirical second-moment matrix $\widehat{\mathbf R}_q$, this error is equivalently given by the sum of the discarded eigenvalues.

Next, let $\widehat Q_{N_n}^{(r)}(z)$ denote the best representation of
$\widetilde Q_{\Delta,M}^{(r)}(z)$ using the $N_n=MK$ configured recurrent basis functions. The corresponding empirical recurrent-basis approximation error is
\begin{equation}
\label{eq:eps_rational_fit}
\widehat{\varepsilon}_{\mathrm{RP}}^{2}(M,K)
:=
\frac{1}{N_{\mathrm{obs}}}
\sum_{r=1}^{N_{\mathrm{obs}}}
\min_{\mathbf W_{\mathrm{out,s}}^{(r)}}
\left\|
\widetilde Q_{\Delta,M}^{(r)}(z)
-
\widehat Q_{N_n}^{(r)}(z)
\right\|_2^2 ,
\end{equation}
where the recurrent basis functions are fixed by the configured
$\mathbf W_{\mathrm{res}}$ and $\mathbf W_{\mathrm{in}}$, while
$\mathbf w_{\mathrm{IIR}}^{(r)}$ contains their realization-specific linear-combination coefficients.

The total empirical recurrent approximation error with respect to the original IIR target is
\begin{equation}
\label{eq:eps_iir_bound}
\widehat{\varepsilon}_{\mathrm{IIR}}^{2}(N_n)
:=
\frac{1}{N_{\mathrm{obs}}}
\sum_{r=1}^{N_{\mathrm{obs}}}
\min_{\mathbf W_{\mathrm{out,s}}^{(r)}}
\left\|
Q_{\Delta}^{(r)}(z)
-
\widehat Q_{N_n}^{(r)}(z)
\right\|_2^2 .
\end{equation}
By the triangle inequality followed by the empirical root-mean-square inequality,
\begin{equation}
\label{eq:eps_iir_decomposition_squared}
\widehat{\varepsilon}_{\mathrm{IIR}}^{2}(N_n)
\leq
\left[
\sqrt{
\widehat{\varepsilon}_{\mathrm{PCA}}^{2}(M)
}
+
\sqrt{
\widehat{\varepsilon}_{\mathrm{RP}}^{2}(M,K)
}
\right]^2 .
\end{equation}
Hence, the recurrent-branch approximation error consists of the PCA truncation error due to retaining only $M$ dominant modes and the additional error incurred when representing the retained low-dimensional target using a finite configured recurrent basis.

For the FIR branch, recall from Proposition~\ref{prop:mmse_iir_fir_structure} that, for realization $r$,
\begin{equation}
W_{\mathrm c}^{(\Delta,r)}(z)
=
Q_{\Delta}^{(r)}(z)
+
P_{\Delta}^{(r)}(z^{-1}),
\end{equation}
where
\begin{equation}
P_{\Delta}^{(r)}(z^{-1})
=
\sum_{d=0}^{D_{\Delta}-1}
f_{d,\Delta}^{(r)}z^{-d}
\end{equation}
is the finite causal FIR component. The $N_w$-tap skip branch realizes
\begin{equation}
\label{eq:F_hat_Nw}
\widehat P_{N_w}^{(r)}(z^{-1})
=
\sum_{d=0}^{N_w-1}
\widehat f_d^{(r)}z^{-d}.
\end{equation}
Its empirical approximation error is defined as
\begin{equation}
\label{eq:eps_fir_bound}
\widehat{\varepsilon}_{\mathrm{FIR}}^{2}(N_w,\Delta)
:=
\frac{1}{N_{\mathrm{obs}}}
\sum_{r=1}^{N_{\mathrm{obs}}}
\min_{\mathbf W_{\mathrm{out,w}}^{(r)}}
\left\|
P_{\Delta}^{(r)}(z^{-1})
-
\widehat P_{N_w}^{(r)}(z^{-1})
\right\|_2^2 .
\end{equation}

Since the skip branch spans the complete causal delay basis
$\{1,z^{-1},\ldots,z^{-(N_w-1)}\}$, the first $N_w$ FIR coefficients can be represented exactly. Therefore, for $N_w<D_{\Delta}$,
\begin{equation}
\label{eq:eps_causal_tail}
\widehat{\varepsilon}_{\mathrm{FIR}}^{2}(N_w,\Delta)
=
\frac{1}{N_{\mathrm{obs}}}
\sum_{r=1}^{N_{\mathrm{obs}}}
\left\|
\sum_{d=N_w}^{D_{\Delta}-1}
f_{d,\Delta}^{(r)}z^{-d}
\right\|_2^2 .
\end{equation}
If $N_w\geq D_{\Delta}$, the FIR component lies entirely in the skip space and
\begin{equation}
\widehat{\varepsilon}_{\mathrm{FIR}}^{2}(N_w,\Delta)=0.
\end{equation}

\begin{proposition}[Empirical mixed IIR--FIR approximation bound]
\label{prop:mixed_iir_fir_approximation_bound}
For the $N_{\mathrm{obs}}$ channel realizations used to construct the configured skip-RNN, the minimum empirical mean-squared $L_2$ approximation error satisfies
\begin{align}
\label{eq:wesn_approx_bound}
&
\frac{1}{N_{\mathrm{obs}}}
\sum_{r=1}^{N_{\mathrm{obs}}}
\min_{\mathbf W_{\mathrm{out}}^{(r)}}
\left\|
G_{\mathrm{SR}}^{(r)}(z)
-
W_{\mathrm c}^{(\Delta,r)}(z)
\right\|_2^2
\nonumber\\
&\qquad\leq
\left[
\sqrt{
\widehat{\varepsilon}_{\mathrm{IIR}}^{2}(N_n)
}
+
\sqrt{
\widehat{\varepsilon}_{\mathrm{FIR}}^{2}(N_w,\Delta)
}
\right]^2 .
\end{align}
\end{proposition}

\begin{proof}
See Appendix~\ref{append_prop_mixed_iir_fir_approximation_bound}.
\end{proof}

The bound separates the approximation loss into two architectural contributions: the error associated with representing the IIR target using the finite configured recurrent basis and the truncation error associated with the finite skip window. These empirical errors characterize the representation capability of the configured receiver over the observed channel ensemble and are distinct from the finite-pilot estimation error incurred during online readout adaptation.

\subsection{Online Readout Adaptation with Prior FIR weights}
\label{subsec:weight_configuration_online_adaptation}

The recurrent parameters of skip-RNN $\left\{\mathbf W_{\mathrm{res}},\mathbf W_{\mathrm{in}}\right\}$ determine the fixed IIR basis, while $N_w$ determines the complete FIR delay basis. The ensemble-averaged matrix $\mathbf W_{\mathrm{out,w}}^{(0)}$ does not change the FIR basis itself; instead, it provides a statistically informed initialization and regularization center for the FIR readout.

After offline configuration, $\mathbf W_{\mathrm{res}}$ and $\mathbf W_{\mathrm{in}}$ remain fixed. Within each TTI, only the readout is adapted from the available reference symbols. Partition it as
\begin{equation}
\mathbf W_{\mathrm{out}}
=
\begin{bmatrix}
\mathbf W_{\mathrm{out,s}}
&
\mathbf W_{\mathrm{out,w}}
\end{bmatrix},
\end{equation}
where the first block combines the recurrent features and the second combines the FIR skip features.

Let $\mathbf Z$ denote the skip-augmented feature matrix and $\mathbf X$ the corresponding known target-symbol matrix. The realization-specific readout is obtained from
\begin{align}
\label{eq:readout_prior}
\widehat{\mathbf W}_{\mathrm{out}}
=
\operatorname*{arg\,min}_{\mathbf W_{\mathrm{out}}}
\;&
\left\|
\mathbf W_{\mathrm{out}}\mathbf Z-\mathbf X
\right\|_F^2 +
\beta
\left\|
\mathbf W_{\mathrm{out,s}}
\right\|_F^2
\nonumber\\
&+
\alpha
\left\|
\mathbf W_{\mathrm{out,w}}
-
\mathbf W_{\mathrm{out,w}}^{(0)}
\right\|_F^2.
\end{align}

In Section~\ref{subsec:low_dim_iir}, for each realization $r$, the causal LMMSE filter is decomposed as in \eqref{eq:causal_mmse_realization_decomp}, yielding the finite-delay coefficients
$\{\mathbf F_{d,\Delta}^{(r)}\}_{d=0}^{N_w-1}$
of $P_{\Delta}^{(r)}(z^{-1})$. The prior FIR readout is then the ensemble mean of these taps,
\begin{equation}
\mathbf W_{\mathrm{out,w}}^{(0)}
=
\begin{bmatrix}
\bar{\mathbf F}_0 &
\bar{\mathbf F}_1 &
\cdots &
\bar{\mathbf F}_{N_w-1}
\end{bmatrix};
\;
\bar{\mathbf F}_d
=
\frac{1}{N_{\mathrm{obs}}}
\sum_{r=1}^{N_{\mathrm{obs}}}
\mathbf F_{d,\Delta}^{(r)}.\nonumber
\end{equation}
This prior centers the regularized readout in \eqref{eq:readout_prior} on the statistically expected FIR correction for NMP channels, while the per-TTI least-squares fit adapts $\mathbf W_{\mathrm{out,w}}$ from $\mathbf W_{\mathrm{out,w}}^{(0)}$ by the observed pilots.

Then the closed-form solution is
\begin{equation}
\label{eq:readout_closed_form_prior}
\widehat{\mathbf W}_{\mathrm{out}}
=
\left(
\mathbf X\mathbf Z^H
+
\mathbf W_0
\boldsymbol\Lambda_{\mathrm{reg}}
\right)
\left(
\mathbf Z\mathbf Z^H
+
\boldsymbol\Lambda_{\mathrm{reg}}
\right)^{-1},
\end{equation}

where, $\boldsymbol\Lambda_{\mathrm{reg}} = \operatorname{blkdiag} \left( \beta\mathbf I_{N_n}, \alpha\mathbf I_{N_rN_w} \right)$, and $\mathbf W_0 = [\mathbf 0 \quad \mathbf W_{\mathrm{out,w}}^{(0)}]$. 

Weight configuration determines a channel-informed IIR basis and an FIR coefficient prior from the channel ensemble. Online readout training does not alter the
configured poles, spatial input modes, or FIR delay basis; it adapts only their realization-specific linear combination using the reference symbols available in the current TTI.

\subsection{Complexity Analysis}
\label{subsec:complexity}
The computational complexity of baseline LMMSE detectors has been thoroughly analyzed in our prior works~\cite{zhou2020learning, xu2022rc, xu2023detect, 10974467}. We focus on the dominant linear-algebra operations like matrix multiplications, pseudo-inverses and eigen decomposition, since matrix additions and element-wise operations are negligible at the scales considered. We adopt the standard complexity model in which multiplying an $m\times n$ matrix by an $n\times k$ matrix incurs $\mathcal{O}(mnk)$ operations, while computing the pseudo-inverse of an $m\times n$ matrix (with $m\le n$) via singular value decomposition (SVD) requires $\mathcal{O}(mn^2)$ operations~\cite{xu2022rc}. 

The total complexity of computing configured weights is the sum of the cost for all the steps which is dominated by the eigen decomposition cost $\mathcal{O}(N_{\mathrm{sc}}^{3})$ in the worst case. Since the above configuration can be computed once and reused for same $E_b/N_o$, its amortized per-iteration overhead over $N_{\mathrm{run}}$ runtime iterations is $\mathrm{FLOPs}_{\mathrm{cfg}}/N_{\mathrm{run}}$, which becomes negligible for long simulations/training runs.

After the weights are configured, the training and inference complexity of the neural receiver remains identical to that of a randomly initialized model. The overall training process involves two main steps: (i) the state-update cost and (ii) the readout weight estimation cost.
Generating reservoir states over $N_{\mathrm{train}}$ samples requires $\mathcal{O}\!\big((N_n+N_r)N_nN_{\mathrm{train}}\big),$ while estimating the readout weights ${\mathbf{W}}_{\mathrm{out}}$ with extended states via least-squares/pseudo-inverse incurs
\begin{equation}
\mathcal{O}\!\big(N_{\mathrm{train}}^2(N_n+N_r N_w) + N_tN_{\mathrm{train}}(N_n+N_r N_w)\big). \nonumber
\end{equation}
Combining both terms, the overall training complexity is
\begin{equation}
\mathcal{O}\!\big(N_{train}((N_n+N_r)N_n + (N_n+N_r N_w)(N_{\mathrm{train}}+N_t))\big). \nonumber
\end{equation}

At test time, the forward pass involves the state update and the output computation using readout weights. Therefore, the testing complexity of RC scales as
\begin{equation}
\mathcal{O}\!\big(N_{test}(N_n(N_n+N_r)+ N_t(N_n + N_r N_w))\big). \nonumber
\end{equation}

\section{Performance Evaluation}
\label{sec:perform_eval}
In this section, we evaluate the performance of the configured skip-RNN neural receiver in a MIMO-OFDM system and compare it with state-of-the-art model-based receivers and learning-based approaches.
\subsection{Simulation Setup}
\label{subsec:system_model}

\begin{figure}[!ht]
\centering
\includegraphics[width=\linewidth]{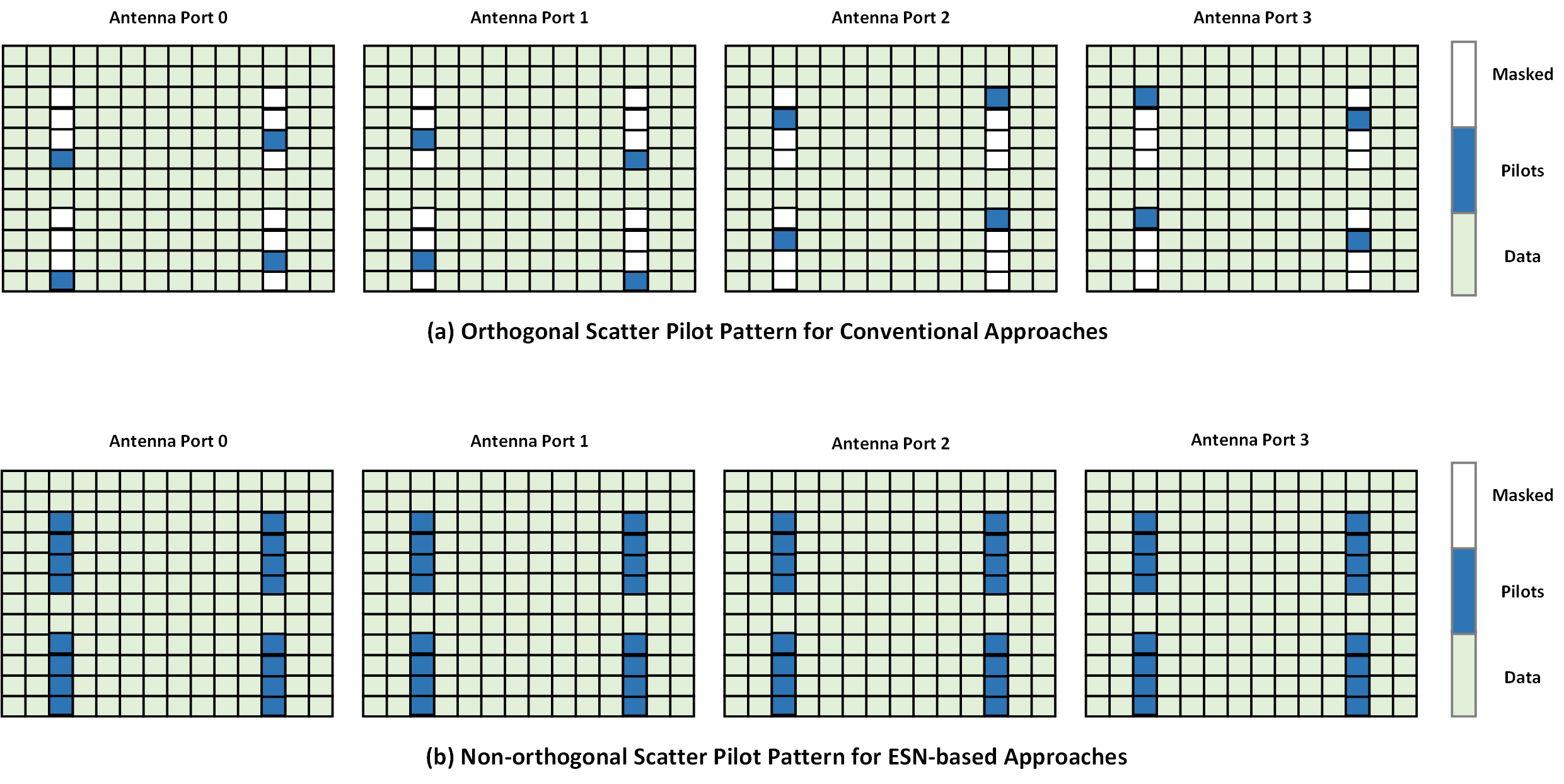}    
\caption{Scatter Pilot Pattern for $4\times 4$ MIMO-OFDM in one resource block}
\label{fig:pilot_pattern}
\end{figure}


We consider a system consisting of a single base station (BS) equipped with four transmit antennas and a single user terminal (UT) equipped with four receive antennas. The user speed is set to 30 km/h with carrier frequency $f_c$ = 3.8 GHz, corresponding to band n78 in 5G NR~\cite{3gpp_carrier_freq}. Both the BS and UT employ dual-polarized antenna arrays compliant with the 3GPP TR~38.901 specification~\cite{3gpp_channel}, using vertical and horizontal polarizations at a carrier frequency of \(3.8\,\mathrm{GHz}\). 

Each TTI comprises 14 OFDM symbols with a subcarrier spacing of \(30\,\mathrm{kHz}\) and an FFT size of 1024. In accordance with the 3GPP NR specification, demodulation reference signals (DMRS) are designed following Mapping Type~A with configuration type~2. 
The resulting pilot pattern forms a scattered, orthogonal structure across frequency, which is repeated over the two DMRS symbols at OFDM indices $2$ and $11$, shown in Fig.~\ref{fig:pilot_pattern}(a). The neural receiver does not require conventional orthogonal pilot design that allocates dedicated pilot resources per spatial layer; rather, it uses a non-orthogonal pilot pattern shown in Fig.~\ref{fig:pilot_pattern}(b), since it does not need to perform joint channel estimation and demodulation. This 3GPP-compliant DMRS configuration enables support for four spatial layers using a combination of CDM groups and time-domain separation while enabling evaluation of the introduced detector under realistic, pilot-scarce scenarios.

\begin{table}[!t]
\centering
\caption{Simulation Parameters}
\label{tab:sim_params}
\vspace{-0.06in}
\scriptsize
\setlength{\tabcolsep}{2.5pt}
\renewcommand{\arraystretch}{1.08}
\begin{tabularx}{\columnwidth}{
@{}>{\raggedright\arraybackslash}X
>{\centering\arraybackslash}p{0.20\columnwidth}
>{\raggedright\arraybackslash}X
>{\centering\arraybackslash}p{0.20\columnwidth}@{}}
\toprule
\textbf{Parameter} & \textbf{Value} &
\textbf{Parameter} & \textbf{Value} \\
\midrule
MIMO configuration & \(4\times4\) &
Spatial layers & \(4\) \\

Carrier frequency & \(3.8\,\mathrm{GHz}\) &
Subcarrier spacing & \(30\,\mathrm{kHz}\) \\

FFT size & \(1024\) &
Effective RBs & \(16\) \\

OFDM symbols & \(14\) &
Pilot symbols & \(2\) \\

Modulation & 16-/64-QAM &
Cyclic-prefix length & \(128\) samples \\

UT speed & \(30\,\mathrm{km/h}\) &
Channel models & CDL-B/C \\

Delay spread & \(100 \; and \; 300\,\mathrm{ns}\) &
Doppler spread & \(100\,\mathrm{Hz}\) \\

Antenna pattern & 3GPP TR 38.901 &
Polarization & Dual (V/H) \\

\bottomrule
\end{tabularx}
\vspace{-0.08in}
\end{table}

\begin{table}[!t]
\centering
\caption{Neural Receiver Parameters}
\label{tab:rc_params}
\vspace{-0.06in}
\scriptsize
\setlength{\tabcolsep}{2.5pt}
\renewcommand{\arraystretch}{1.08}
\begin{tabularx}{\columnwidth}{
@{}>{\raggedright\arraybackslash}X
>{\centering\arraybackslash}p{0.18\columnwidth}
>{\raggedright\arraybackslash}X
>{\centering\arraybackslash}p{0.18\columnwidth}@{}}
\toprule
\textbf{Parameter} & \textbf{Value} &
\textbf{Parameter} & \textbf{Value} \\
\midrule
Recurrent neurons, \(N_n\) & \(32\) &
Sparsity & \(0.4\) \\

Spectral radius & \(0.5\) &
Input scaling & \(0.8\) \\

Window length, \(N_w\) & \(32\) &
Top eigenvectors, \(M\) & \(4\) \\

RP order, \(K\) & \(8\) &
Activation function, \(\sigma(\cdot)\) & \(Id\) \\
\bottomrule
\end{tabularx}
\vspace{-0.08in}
\end{table}

Wireless propagation is modeled using standardized 3GPP NLoS CDL channel models including CDL-B with $100 ns$ delay spread and CDL-C with $300 ns$ delay spread~\cite{3gpp_channel} implemented in Sionna~\cite{sionna}. 
The cyclic prefix length is set to 128 to accommodate channels with large delay spreads, and the simulation is performed over 16 RBs. The detailed simulation parameters are listed in Table~\ref{tab:sim_params}.


\begin{figure*}[t]
\vspace{-0.1in}
    \centering
    \subfloat{
        \includegraphics[width=0.425\linewidth]{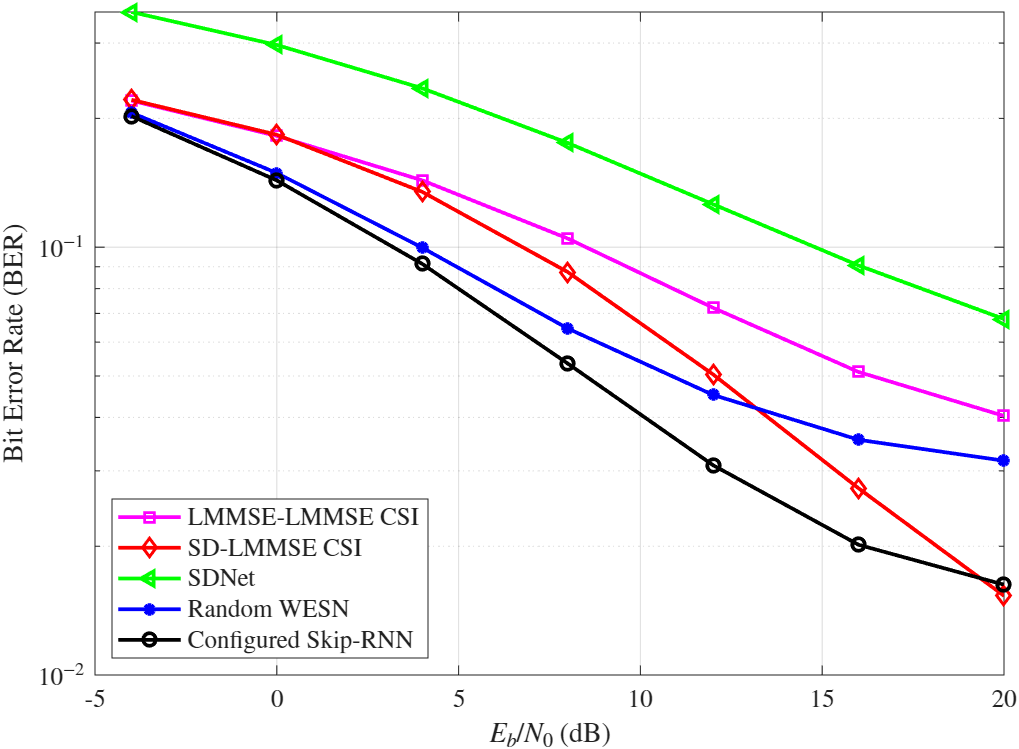}
    }\hfill
    \subfloat{
        \includegraphics[width=0.425\linewidth]{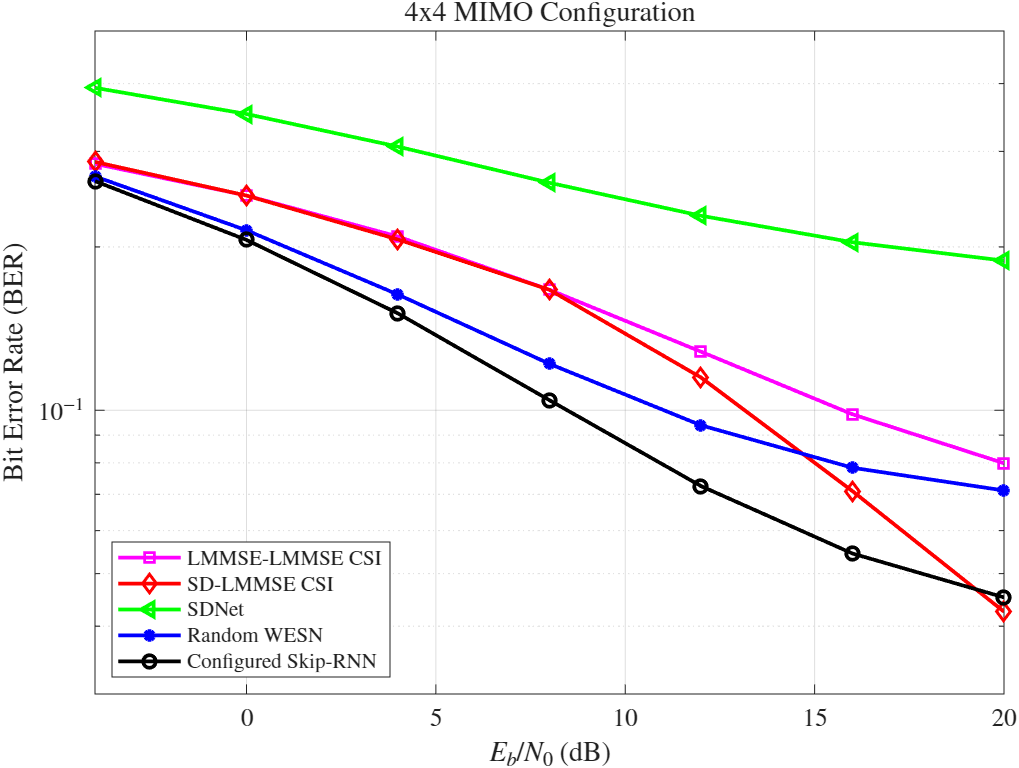}
    }
    \caption{BER versus $E_b/N_0$ for different detectors under CDL-B channel for (a) 16-QAM and (b) 64-QAM.}
    \label{fig:cdlB_ber16_64}
    \vspace{-0.05in}

\end{figure*}

In the experimental setup, the skip-RNN model utilizes $N_n = 32$ recurrent neurons and a window length of $N_w = 32$. For weight configuration, we use 20 slots of estimated frequency domain channel realizations to compute the empirical covariance of the inverse channel responses. The top $M = 4$ eigenvectors are extracted and approximated using RP of order $K = 8$. The complete set of neural receiver parameters is summarized in Table~\ref{tab:rc_params}.

\subsection{Baseline Receive Processing Approach}
We are benchmarking our configured neural receiver against model-based LMMSE and SD receivers and learning-based SDNet and WESN detector with random weights:

\begin{itemize}
\item \textbf{LMMSE + LMMSE-CSI:} This baseline applies two-stage processing: a)LMMSE channel estimation with 2D LMMSE interpolation, and b)LMMSE receive processing. LMMSE channel estimation technique consists of two steps:(i) LS-based estimation of the channel at pilot REs, followed by (ii) 2D LMMSE interpolation across the time–frequency grid to estimate the full channel matrix. The estimated channel is then used for LMMSE-based detection. While computationally demanding due to the covariance estimation and 2D LMMSE interpolation, this pipeline is a strong model-based baseline with a favorable performance-complexity trade-off~\cite{cammerer2023neural}. 

\item \textbf{SD + LMMSE-CSI:} This baseline uses the same LMMSE-based channel estimate as the linear receiver, but replaces LMMSE detection with sphere decoding (SD). For each data RE, SD searches the discrete symbol space to minimize the Euclidean detection metric, providing near-ML detection when the channel estimate is sufficiently accurate~\cite{ghasemmehdi2011faster}.

\item \textbf{SDNet:} SDNet is included as a learning-based MIMO detector that refines an initial model-based symbol estimate. Rather than directly processing the received waveform, the network operates on two intermediate quantities: the zero-forcing (ZF) symbol estimate, $\widehat{\mathbf{x}}_{\mathrm{ZF}}$, and an LS estimate of the channel matrix, $\widehat{\mathbf{H}}_{\mathrm{LS}}$. The real and imaginary components of these quantities are rearranged and combined to form the network input. A sequence of convolutional layers extracts joint information from the coarse symbol estimate and the estimated channel, after which a fully connected output layer produces the refined detected-symbol vector. Unlike our introduced receiver, SDNet relies explicitly on channel estimation and requires offline optimization of its neural-network parameters~\cite{zhang2025efficient}.

\item \textbf{Random WESN:} Reservoir-based online real-time receiver using randomly initialized fixed weights. This model utilizes windowed input and operates on time domain samples and performs direct per-TTI MIMO detection without requiring explicit channel estimation~\cite{zhou2020learning}.

\end{itemize}

\begin{figure}[!ht]
\vspace{-0.15in}
\centering
\includegraphics[width=0.8\linewidth]{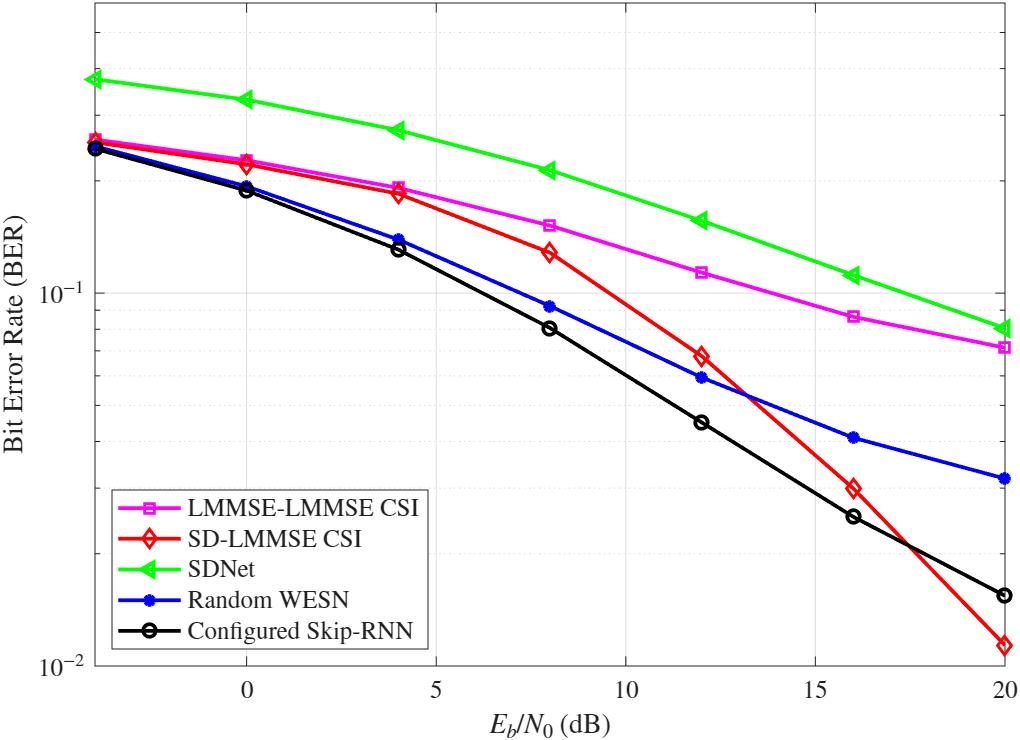}   
\vspace{-0.15in}
\caption{BER performance for 16-QAM under CDL-C channel}
\label{fig:cdlc_16qam}
\end{figure}

\begin{figure}[!ht]
\vspace{-0.15in}
\centering
\includegraphics[width=0.8\linewidth]{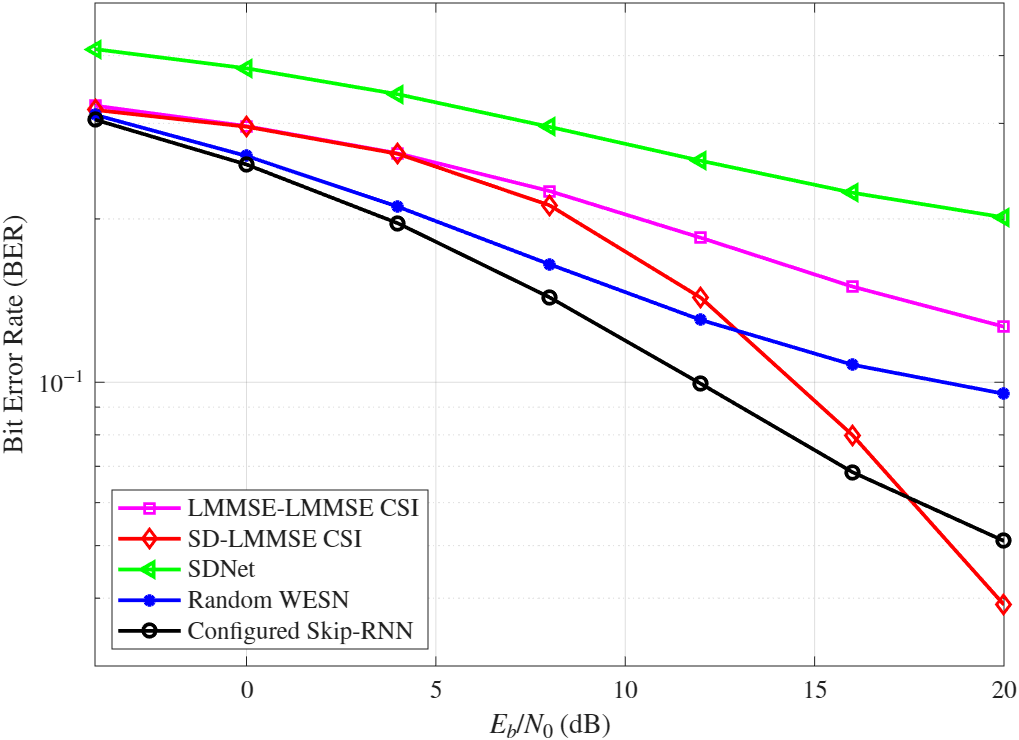}   
\caption{BER performance for 64-QAM under CDL-C channel}
\vspace{-0.15in}
\label{fig:cdlc_64qam}
\end{figure}

\vspace{-0.15in}
\subsection{BER Performance Comparison}
The BER performance of the considered receivers is evaluated over the 3GPP NLoS CDL-B and CDL-C channel models. Figure~\ref{fig:cdlB_ber16_64} presents the results for CDL-B with 16-QAM and 64-QAM, while Figures~\ref{fig:cdlc_16qam} and~\ref{fig:cdlc_64qam} report the corresponding results for CDL-C. Across both channel models and modulation orders, the proposed configured skip-RNN achieves the lowest BER among the evaluated receivers, followed by the Random WESN. This consistent improvement over the randomly configured reservoir demonstrates the benefit of adapting the recurrent structure to the channel characteristics rather than relying on randomly generated reservoir weights.

Among the receivers that rely on estimated CSI, SDNet exhibits the highest BER, while LMMSE + LMMSE-CSI performs better over most of the evaluated SNR range. In the high-SNR regime, SD + LMMSE-CSI outperforms LMMSE + LMMSE-CSI because its nonlinear search over the discrete symbol space more closely realizes the maximum-likelihood (ML) detection criterion. At low SNR, this advantage diminishes as channel-estimation errors distort the detection metric. At sufficiently high SNR, e.g., $E_b/N_0=20$~dB, where the estimated CSI becomes more accurate, SD also outperforms the configured skip-RNN. This behavior is expected since the SD algorithm solves the ML detection problem for a given channel realization and discrete symbol alphabet~\cite{vikalo2005sphere}.

The performance difference between the configured skip-RNN and the Random WESN also becomes smaller in the low-$E_b/N_0$ regime. Under severe noise, the channel information used during weight configuration becomes less reliable, limiting the benefit of channel-informed weight configuration. In high SNR, the configured reservoir provides a more suitable recurrent representation of the underlying channel dynamics, resulting in better BER performance over the Random WESN. Overall, the results indicate that the proposed configuration improves the robustness of our neural receiver across different NLoS delay profiles and higher-order modulation schemes without requiring an explicit CSI estimation.

\section{Conclusion}
\label{sec:conclusion}

This work introduced an interpretable skip-RNN receiver for MIMO-OFDM under mixed-phase NLoS channels. We showed that the skip-RNN operates as a causal deconvolution filter whose recurrent and skip branches realize complementary IIR and FIR basis functions, matching the structure of the causal LMMSE channel inverse filter. This interpretation motivated a channel-informed configuration framework in which historical estimated channel statistics are used to configure the recurrent dynamics, determine the FIR memory, and construct a statistical readout prior for prior-centered regularized online adaptation. An approximation-error analysis further characterized the representational capability of the configured receiver. Simulation results over standardized 3GPP NLoS mixed-phase channels demonstrated that the weight configuration improves detection performance over conventional model-based and learning-based baselines under limited pilot overhead.

\vspace{-0.15in}

\section*{Acknowledgment} 

The authors would like to thank Shashank Jere for insightful discussions and for proofreading the manuscript.

\vspace{-0.1in}

\appendices

\section{PROOF OF PROPOSITION~\ref{prop:mmse_iir_fir_structure}}
\label{append_prop_mmse_iir_fir_structure}

\begin{proof}
From \eqref{eq:causal_mmse_iir_fir_cascade}, the causal LMMSE channel inverse is
\begin{equation}
W_{\mathrm{c}}^{(\Delta)}(z)
=
\frac{F_{\Delta}(z)}{\Phi(z)},\nonumber
\end{equation}
where $F_{\Delta}(z) = \sum_{k=0}^{\Delta} r[-k]z^{-(\Delta-k)}$ is a finite-degree polynomial in $z^{-1}$.

Since $\Phi(z)$ is the canonical minimum-phase spectral factor, $\Phi^{-1}(z)$ is causal and BIBO-stable. By the rationality assumption, $W_{\mathrm{c}}^{(\Delta)}(z)$ is therefore a causal, stable, rational transfer function. Its finite poles are the poles of $\Phi^{-1}(z)$ that are not cancelled by zeros of $F_{\Delta}(z)$. 

Polynomial division separates $W_{\mathrm{c}}^{(\Delta)}(z)$ into a finite polynomial part and a strictly proper rational part:
\begin{equation}
\label{eq:proof_polynomial_division}
W_{\mathrm{c}}^{(\Delta)}(z)
=
P_{\Delta}(z^{-1})
+
Q_{\Delta}(z),
\end{equation}
where $P_{\Delta}(z^{-1})$ is a finite-degree polynomial in $z^{-1}$ and $Q_{\Delta}(z)$ is strictly proper.

Under the distinct-pole assumption, the strictly proper component admits the partial-fraction expansion
\begin{equation}
Q_{\Delta}(z)
=
\sum_{m=1}^{M}
\frac{c_{m,\Delta}}
{1-p_mz^{-1}},
\qquad
|p_m|<1. \nonumber
\end{equation}

$P_{\Delta}(z^{-1})$ can be written as
\begin{equation}
P_{\Delta}(z^{-1})
=
\sum_{d=0}^{D_{\Delta}-1}
f_{d,\Delta}z^{-d},\nonumber
\end{equation}
which is an FIR transfer function. Substitution into~\eqref{eq:proof_polynomial_division} gives
\begin{equation}
W_{\mathrm{c}}^{(\Delta)}(z)
=
\sum_{m=1}^{M}
\frac{c_{m,\Delta}}
{1-p_mz^{-1}}
+
\sum_{d=0}^{D_{\Delta}-1}
f_{d,\Delta}z^{-d},
\qquad
|p_m|<1, \nonumber
\end{equation}
which proves \eqref{eq:unified_parallel_form}.
\end{proof}

\section{Proof of Proposition~\ref{prop:mixed_iir_fir_approximation_bound}}
\label{append_prop_mixed_iir_fir_approximation_bound}
\begin{proof}
For channel realization $r$, the causal LMMSE channel inverse admits the parallel decomposition
\begin{equation}
W_{\mathrm c}^{(\Delta,r)}(z)
=
Q_{\Delta}^{(r)}(z)
+
P_{\Delta}^{(r)}(z^{-1}),
\end{equation}
while the skip-RNN realizes
\begin{equation}
G_{\mathrm{SR}}^{(r)}(z)
=
\widehat Q_{N_n}^{(r)}(z)
+
\widehat P_{N_w}^{(r)}(z^{-1}),
\end{equation}
where $\widehat Q_{N_n}^{(r)}(z)$ and
$\widehat P_{N_w}^{(r)}(z^{-1})$ are controlled by the recurrent and skip readout blocks
$\mathbf W_{\mathrm{out,s}}^{(r)}$ and
$\mathbf W_{\mathrm{out,w}}^{(r)}$, respectively.

Define the minimum branchwise approximation errors as
\begin{align}
e_{\mathrm{IIR}}^{(r)}
&:=
\min_{\mathbf W_{\mathrm{out,s}}^{(r)}}
\left\|
Q_{\Delta}^{(r)}(z)
-
\widehat Q_{N_n}^{(r)}(z)
\right\|_2,
\label{eq:proof_iir_error}\\
e_{\mathrm{FIR}}^{(r)}
&:=
\min_{\mathbf W_{\mathrm{out,w}}^{(r)}}
\left\|
P_{\Delta}^{(r)}(z^{-1})
-
\widehat P_{N_w}^{(r)}(z^{-1})
\right\|_2.
\label{eq:proof_fir_error}
\end{align}

Because the two branches correspond to disjoint blocks of
\begin{equation}
\mathbf W_{\mathrm{out}}^{(r)}
=
\left[
\mathbf W_{\mathrm{out,s}}^{(r)}
\;\;
\mathbf W_{\mathrm{out,w}}^{(r)}
\right],
\end{equation}
their coefficients can be selected independently. Let the two branchwise minimizers in
\eqref{eq:proof_iir_error}--\eqref{eq:proof_fir_error} be used together as a feasible joint readout. Then
\begin{align}
\min_{\mathbf W_{\mathrm{out}}^{(r)}}
\left\|
G_{\mathrm{SR}}^{(r)}(z)
-
W_{\mathrm c}^{(\Delta,r)}(z)
\right\|_2 \leq
e_{\mathrm{IIR}}^{(r)}
+
e_{\mathrm{FIR}}^{(r)}.
\label{eq:proof_per_realization_bound}
\end{align}

Squaring \eqref{eq:proof_per_realization_bound} and averaging over the
$N_{\mathrm{obs}}$ observed channel realizations gives
\begin{align}
&
\frac{1}{N_{\mathrm{obs}}}
\sum_{r=1}^{N_{\mathrm{obs}}}
\min_{\mathbf W_{\mathrm{out}}^{(r)}}
\left\|
G_{\mathrm{SR}}^{(r)}(z)
-
W_{\mathrm c}^{(\Delta,r)}(z)
\right\|_2^2
\nonumber\\
&\quad\leq
\frac{1}{N_{\mathrm{obs}}}
\sum_{r=1}^{N_{\mathrm{obs}}}
\left(
e_{\mathrm{IIR}}^{(r)}
+
e_{\mathrm{FIR}}^{(r)}
\right)^2.
\label{eq:proof_empirical_sum}
\end{align}

Expanding the right-hand side and applying the Cauchy--Schwarz inequality to the cross term yields

\begin{align}
\frac{1}{N_{\mathrm{obs}}}
\sum_{r=1}^{N_{\mathrm{obs}}}
e_{\mathrm{IIR}}^{(r)}
e_{\mathrm{FIR}}^{(r)}
&\leq
\sqrt{
\frac{1}{N_{\mathrm{obs}}}
\sum_{r=1}^{N_{\mathrm{obs}}}
\left(e_{\mathrm{IIR}}^{(r)}\right)^2
}
\nonumber\\
&\quad\times
\sqrt{
\frac{1}{N_{\mathrm{obs}}}
\sum_{r=1}^{N_{\mathrm{obs}}}
\left(e_{\mathrm{FIR}}^{(r)}\right)^2
}.
\end{align}

By the definitions of the empirical branchwise approximation errors,
\begin{equation}
\widehat{\varepsilon}_{\mathrm{IIR}}^{2}(N_n)
=
\frac{1}{N_{\mathrm{obs}}}
\sum_{r=1}^{N_{\mathrm{obs}}}
\left(e_{\mathrm{IIR}}^{(r)}\right)^2,
\end{equation}
and
\begin{equation}
\widehat{\varepsilon}_{\mathrm{FIR}}^{2}(N_w,\Delta)
=
\frac{1}{N_{\mathrm{obs}}}
\sum_{r=1}^{N_{\mathrm{obs}}}
\left(e_{\mathrm{FIR}}^{(r)}\right)^2.
\end{equation}
Therefore,
\begin{align}
&
\frac{1}{N_{\mathrm{obs}}}
\sum_{r=1}^{N_{\mathrm{obs}}}
\left(
e_{\mathrm{IIR}}^{(r)}
+
e_{\mathrm{FIR}}^{(r)}
\right)^2
\nonumber\\
&\quad\leq
\left[
\sqrt{
\widehat{\varepsilon}_{\mathrm{IIR}}^{2}(N_n)
}
+
\sqrt{
\widehat{\varepsilon}_{\mathrm{FIR}}^{2}(N_w,\Delta)
}
\right]^2.
\end{align}

Combining this result with \eqref{eq:proof_empirical_sum} gives
\begin{align}
&
\frac{1}{N_{\mathrm{obs}}}
\sum_{r=1}^{N_{\mathrm{obs}}}
\min_{\mathbf W_{\mathrm{out}}^{(r)}}
\left\|
G_{\mathrm{SR}}^{(r)}(z)
-
W_{\mathrm c}^{(\Delta,r)}(z)
\right\|_2^2
\nonumber\\
&\quad\leq
\left[
\sqrt{
\widehat{\varepsilon}_{\mathrm{IIR}}^{2}(N_n)
}
+
\sqrt{
\widehat{\varepsilon}_{\mathrm{FIR}}^{2}(N_w,\Delta)
}
\right]^2,
\end{align}
which proves~\eqref{eq:wesn_approx_bound}.
\end{proof}

\bibliography{Bibliography/IEEEabrv, Bibliography/sample}

@STRING{IEEE_J_WCOM       = "{IEEE} Trans. Wireless Commun."}

@article{ghasemmehdi2011faster,
  title={Faster recursions in sphere decoding},
  author={Ghasemmehdi, Arash and Agrell, Erik},
  journal={IEEE Trans. Inf. Theory},
  volume={57},
  number={6},
  pages={3530--3536},
  year={2011},
  publisher={IEEE}
}

@book{kailath2000linear,
  author    = {Kailath, T. and Sayed, A. H. and Hassibi, B.},
  title     = {Linear Estimation},
  publisher = {Prentice Hall},
  address   = {Upper Saddle River, NJ},
  year      = {2000}
}

@book{sayed2003fundamentals,
  author    = {Sayed, A. H.},
  title     = {Fundamentals of Adaptive Filtering},
  publisher = {Wiley-IEEE Press},
  year      = {2003}
}

@article{gustavsen1999rational,
  title={Rational approximation of frequency domain responses by vector fitting},
  author={Gustavsen, Bjorn and Semlyen, Adam},
  journal={IEEE Trans. Power Del.},
  volume={14},
  number={3},
  pages={1052--1061},
  year={1999},
  publisher={IEEE}
}

@article{ye2017power,
  title={Power of deep learning for channel estimation and signal detection in {OFDM} systems},
  author={Ye, Hao and Li, Geoffrey Ye and Juang, Biing-Hwang},
  journal={IEEE Wireless Commun. Let.},
  volume={7},
  number={1},
  pages={114--117},
  year={2017},
  publisher={IEEE}
}

@article{honkala2021deeprx,
  title={DeepRx: Fully convolutional deep learning receiver},
  author={Honkala, Mikko and Korpi, Dani and Huttunen, Janne MJ},
  journal={IEEE Trans. Wireless Commun.},
  volume={20},
  number={6},
  pages={3925--3940},
  year={2021},
  publisher={IEEE}
}

@article{zhao2021deep,
  title={Deep-waveform: A learned {OFDM} receiver based on deep complex-valued convolutional networks},
  author={Zhao, Zhongyuan and Vuran, Mehmet Can and Guo, Fujuan and Scott, Stephen D},
  journal={IEEE J. Sel. Areas Commun.},
  volume={39},
  number={8},
  pages={2407--2420},
  year={2021},
  publisher={IEEE}
}

@inproceedings{liao2019deep,
  title={Deep neural network symbol detection for millimeter wave communications},
  author={Liao, Yun and Farsad, Nariman and Shlezinger, Nir and Eldar, Yonina C and Goldsmith, Andrea J},
  booktitle={2019 IEEE Global Commun. Conf. (GLOBECOM)},
  pages={1--6}
}

@inproceedings{lyu2020deep,
  title={Deep neural network-based symbol detection for highly dynamic channels},
  author={Lyu, Xuantao and Feng, Wei and Ge, Ning},
  booktitle={2020 IEEE Global Commun. Conf. (GLOBECOM)},
  pages={1--6}
}

@article{yi2020deep,
  title={Deep learning for joint channel estimation and signal detection in {OFDM} systems},
  author={Yi, Xuemei and Zhong, Caijun},
  journal={IEEE Commun. Let.},
  volume={24},
  number={12},
  pages={2780--2784},
  year={2020},
  publisher={IEEE}
}

@article{khani2020adaptive,
  title={Adaptive neural signal detection for massive {MIMO}},
  author={Khani, Mehrdad and Alizadeh, Mohammad and Hoydis, Jakob and Fleming, Phil},
  journal={IEEE Trans. Wireless Commun.},
  volume={19},
  number={8},
  pages={5635--5648},
  year={2020},
  publisher={IEEE}
}

@article{samuel2019learning,
  title={Learning to detect},
  author={Samuel, Neev and Diskin, Tzvi and Wiesel, Ami},
  journal={IEEE Trans. Signal Process.},
  volume={67},
  number={10},
  pages={2554--2564},
  year={2019},
  publisher={IEEE}
}

@inproceedings{he2018model,
  title={A model-driven deep learning network for {MIMO} detection},
  author={He, Hengtao and Wen, Chao-Kai and Jin, Shi and Li, Geoffrey Ye},
  booktitle={2018 IEEE Global Conf. on Signal and Inf. Process (GlobalSIP)},
  pages={584--588}
}

@book{goodfellow2016deep,
  title={Deep learning},
    author={Ian Goodfellow and Yoshua Bengio and Aaron Courville},
    publisher={MIT Press},
    note={\url{http://www.deeplearningbook.org}},
    year={2016}
}

@article{zhang2025efficient,
  title={Efficient Deployment of Deep {MIMO} Detection Using Learngene},
  author={Zhang, Jinya and Guo, Jiajia and Li, Xiangyi and Wen, Chao-Kai and Geng, Xin and Jin, Shi},
  journal={IEEE Trans. Wireless Commun.},
  year={2026},
  volume={25},
  number={},
  pages={4405-4418}
}

@article{liu_magazine,
  title={A Universal Neural Receiver that Learns at the Speed of Wireless},
  author = {Liu, Lingjia and Zheng, Lizhong and Yi, Yang and Calderbank, Robert},
  journal={IEEE Commun. Mag.},
  year={2026}
}

@ARTICLE{letter_ummay,
  author={Khan, Ummay Sumaya and Liu, Lingjia and Jere, Shashank and Zheng, Lizhong and Yi, Yang},
  journal={IEEE Wireless Commun. Let.}, 
  title={Configuring {RNN} Weights for {MIMO}-{OFDM} Receive Processing: Informing {RNN} With Domain Knowledge}, 
  year={2026},
  volume={15},
  number={},
  pages={1712-1716},
  doi={10.1109/LWC.2025.3636058}}

@article{xu2023detect,
  title={Detect to learn: Structure learning with attention and decision feedback for {MIMO}-{OFDM} receive processing},
  author={Xu, Jiarui and Li, Lianjun and Zheng, Lizhong and Liu, Lingjia},
  journal={IEEE Trans. Commun.},
  volume={72},
  number={1},
  pages={146--161},
  year={2023},
  publisher={IEEE}
}

@article{jaeger2001echo,
  title={The “echo state” approach to analysing and training recurrent neural networks-with an erratum note},
  author={Jaeger, Herbert},
  journal={Bonn, Germany: German national research center for information technology gmd technical report},
  volume={148},
  number={34},
  pages={13},
  year={2001},
  publisher={Bonn}
}

@article{zhou2020rcnet,
  title={{RCNet}: Incorporating structural information into deep {RNN} for online {MIMO}-{OFDM} symbol detection with limited training},
  author={Zhou, Zhou and Liu, Lingjia and Jere, Shashank and Zhang, Jianzhong and Yi, Yang},
  journal={IEEE Trans. Wireless Commun.},
  volume={20},
  number={6},
  pages={3524--3537},
  year={2021},
  publisher={IEEE}
}

@article{zhou2020learning,
  title={Learning for detection: {MIMO}-{OFDM} symbol detection through downlink pilots},
  author={Zhou, Zhou and Liu, Lingjia and Chang, Hao-Hsuan},
  journal={IEEE Trans. Wireless Commun.},
  volume={19},
  number={6},
  pages={3712--3726},
  year={2020},
  publisher={IEEE}
}

@article{xu2022rc,
  title={{RC-Struct}: A structure-based neural network approach for {MIMO}-{OFDM} detection},
  author={Xu, Jiarui and Zhou, Zhou and Li, Lianjun and Zheng, Lizhong and Liu, Lingjia},
  journal={IEEE Trans. Wireless Commun.},
  volume={21},
  number={9},
  pages={7181--7193},
  year={2022},
  publisher={IEEE}
}

@article{jere2023channel,
  title={Channel equalization through reservoir computing: A theoretical perspective},
  author={Jere, Shashank and Safavinejad, Ramin and Zheng, Lizhong and Liu, Lingjia},
  journal={IEEE Wireless Commun. Let.},
  volume={12},
  number={5},
  pages={774--778},
  year={2023},
  publisher={IEEE}
}

@ARTICLE{10974467,
  author={Jere, Shashank and Zheng, Lizhong and Said, Karim and Liu, Lingjia},
  journal=IEEE_J_WCOM, 
  title={{Toward xAI: Configuring {RNN} Weights Using Domain Knowledge for {MIMO} Receive Processing}}, 
  year={2025},
  volume={24},
  number={9},
  pages={7581-7597},
  doi={10.1109/TWC.2025.3561242}}

@article{jere2024universal,
  title={Universal approximation of linear time-invariant ({LTI}) systems through {RNN}s: Power of randomness in reservoir computing},
  author={Jere, Shashank and Zheng, Lizhong and Said, Karim and Liu, Lingjia},
  journal={IEEE J. Sel. Topics Signal Process.},
  year={2024},
  publisher={IEEE}
}

@software{sionna,
 title = {Sionna},
 author = {Hoydis, Jakob and Cammerer, Sebastian and {Ait Aoudia}, Fayçal and Nimier-David, Merlin and Maggi, Lorenzo and Marcus, Guillermo and Vem, Avinash and Keller, Alexander},
 note = {https://nvlabs.github.io/sionna/},
 year = {2022},
 version = {1.1.0}
}

@book{oppenheim1999discrete,
  title={Discrete-time signal processing},
  author={Oppenheim, Alan V},
  year={1999},
  publisher={Pearson Education India}
}

@Techreport{3gpp_channel,
author = {3GPP}, 
title = {{Study on channel model for frequencies from 0.5 to 100 GHz}}, 
language = {en},
institution = {{3rd Generation Partnership Project (3GPP)}}, 
number = {{TR 38.901}} ,
note = {{V19.1.0}} ,
type= {{Technical Report (TR),}},
month = {Sep.}, 
year ={2025}, 
url = {{https://www.3gpp.org/ftp/Specs/archive/38_series/38.901/}}
}

@Techreport{3gpp_pilot,
author = {{3GPP}}, 
title = {{NR; Physical channels and modulation}}, 
language = {en},
institution = {{3rd Generation Partnership Project (3GPP)}}, 
number = {{TS 38.211}} ,
note = {{V19.1.0}} ,
type= {{Technical Specification (TS),}},
month = {Sep.}, 
year ={2025}, 
url = {{https://www.3gpp.org/ftp/Specs/archive/38_series/38.211/}}
}

@Techreport{3gpp_mmse,
author = {{3GPP}}, 
title = {{Study on New Radio Access Technology - Physical Layer Aspects}}, 
language = {en},
institution = {{3rd Generation Partnership Project (3GPP)}}, 
number = {{TR 38.802}} ,
note = {{V14.2.0}} ,
type= {{Technical Report (TR),}},
month = {Sep.}, 
year ={2017}, 
url = {{https://www.3gpp.org/ftp/Specs/archive/38_series/38.802/}}
}

@Techreport{3gpp_carrier_freq,
author = {{3GPP}}, 
title = {{5G; NR; Base Station (BS) radio transmission and reception}}, 
language = {en},
institution = {{3rd Generation Partnership Project (3GPP)}}, 
number = {{TS 38.104}} ,
note = {{V19.2.0}} ,
type= {{Technical Report (TS),}},
month = {Sep.}, 
year ={2025}, 
url = {{https://www.3gpp.org/ftp/Specs/archive/38_series/38.104/}}
}

@article{moir2013inverting,
  title={Inverting Non-minimum Phase {FIR} Transfer Functions},
  author={{T. Moir}},
  journal={American J. of Signal Process.},
  volume={3},
  number={4},
  pages={102--105},
  year={2013}
}

@article{moir2007z,
  title={A z-domain transfer function solution to the non-minimum phase acoustic beamformer},
  author={Moir, TJ},
  journal={Int. J. Sys. Sci.},
  volume={38},
  number={7},
  pages={563--575},
  year={2007},
  publisher={Taylor \& Francis}
}

@inproceedings{cammerer2023neural,
  title={A neural receiver for {5G} {NR} multi-user {MIMO}},
  author={Cammerer, Sebastian and A{\"\i}t Aoudia, Fay{\c{c}}al and Hoydis, Jakob and Oeldemann, Andreas and Roessler, Andreas and Mayer, Timo and Keller, Alexander},
  booktitle={2023 IEEE Globecom Workshops (GC Wkshps)},
  pages={329--334}
}

@article{vikalo2005sphere,
  title={On the sphere-decoding algorithm {II}. Generalizations, second-order statistics, and applications to communications},
  author={Vikalo, Haris and Hassibi, Babak},
  journal={IEEE Trans. Signal Proc.},
  volume={53},
  number={8},
  pages={2819--2834},
  year={2005},
  publisher={IEEE}
}
\bibliographystyle{IEEEtran}

\newpage

\vfill

\end{document}